\RequirePackage{fix-cm}
\documentclass[twocolumn]{svjour3}          % twocolumn
\smartqed  % flush right qed marks, e.g. at end of proof
\usepackage{graphicx}
\usepackage{epsfig, amsmath, amssymb, latexsym, multirow, enumerate, url, natbib, booktabs, xcolor, bbm}
\newcommand\e{{\rm e}}                        
\newcommand\tr{\text{\rm tr}}
\newcommand\diff{\text{\rm d}}

\newcommand\var{\text{\rm var}}
\newcommand\logit{\text{\rm logit}}
\newcommand\KL{\text{\rm KL}}

\newcommand\diag{\text{\rm diag}}
\newcommand\dg{\text{\rm dg}}
\newcommand\VEC{\text{\rm vec}}
\newcommand\vech{v}

\newcommand\VI{\text{VI}}
\newcommand\IW{\text{IW}}
\newcommand\mL{{\mathcal{L}}}
\newcommand\N{{\mathcal{N}}}
\newcommand\G{{\mathcal{G}}}
\newcommand\bone{{\bf 1}}

\usepackage{float} 
\floatstyle{plain}
\newfloat{Algorithm}{thp}{lop}
\floatname{Algorithm}{Algorithm}

\begin{document}
\sloppy

\title{Conditionally structured variational Gaussian approximation with importance weights\thanks{Linda Tan and Aishwarya Bhaskaran are supported by the start-up grant R-155-000-190-133. }
}


\author{Linda S. L. Tan  \and
	Aishwarya Bhaskaran \and
	David J. Nott.
}

\authorrunning{Tan, Bhaskaran and Nott} 

\institute{Linda S. L. Tan \at
Department of Statistics and Applied Probability \\
National University of Singapore \\
Tel.: +65-6516-8473\\
Fax: +65-6872-3939\\
\email{statsll@nus.edu.sg}           %  \\
\and
Aishwarya Bhaskaran \at
Department of Statistics and Applied Probability \\
National University of Singapore \\
Tel.: +65-6601-6229 \\
Fax: +65-6872-3939\\
\email{staai@nus.edu.sg} 
\and
David J. Nott \at
Department of Statistics and Applied Probability \\
National University of Singapore \\
Tel.: +65-6516-2744 \\
Fax: +65-6872-3939\\
\email{standj@nus.edu.sg} 
}

\date{Received: date / Accepted: date}

\maketitle

\begin{abstract}
We develop flexible methods of deriving variational inference for models with complex latent variable structure. By splitting the variables in these models into ``global" parameters and ``local"  latent variables, we define a class of variational approximations that exploit this partitioning and go beyond Gaussian variational approximation. This approximation is motivated by the fact that in many hierarchical models, there are global variance parameters which determine the scale of local latent variables in their posterior conditional on the global parameters. We also consider parsimonious parametrizations by using conditional independence structure, and improved estimation of the log marginal likelihood and variational density using importance weights. These methods are shown to improve significantly on Gaussian variational approximation methods for a similar computational cost. Application of the methodology is illustrated using generalized linear mixed models and state space models.
\end{abstract}

\section{Introduction} \label{sec:Introduction}

In many modern statistical applications, it is necessary to model complex dependent data. In these situations, models which employ observation specific latent variables such as random effects and state space models are widely used because of their flexibility, and Bayesian approaches dealing naturally with the hierarchical structure are attractive in principle. However, incorporating observation specific latent variables leads to a parameter dimension increasing with sample size, and standard Bayesian computational methods can be challenging to implement in very high-dimensional settings. For this reason, approximate inference methods are attractive for these models, both in exploratory settings where many models need to be fitted quickly, as well as in applications involving large datasets where exact methods are infeasible. One of the most common approximate inference paradigms is variational inference \citep{Ormerod2010,blei2017}, which is the approach considered here.

Our main contribution is to consider partitioning the unknowns in a local latent variable model into ``global" parameters and ``local" latent variables, and to suggest ways of structuring the dependence in a variational approximation that match the specification of these models.  We go beyond standard Gaussian approximations by defining the variational approximation sequentially, through a marginal density for the global parameters and a conditional density for local parameters given global parameters. Each term in our approximation is Gaussian, but we allow the conditional covariance matrix for the local parameters to depend on the global parameters, which leads to an approximation that is not jointly Gaussian. We are particularly interested in improved inference on global variance and dependence parameters which determine the scale and dependence structure of local latent variables.  With this objective, we suggest a parametrization of our conditional approximation to the local variables that is well-motivated and respects the exact conditional independence structure in the true posterior distribution. Our approximations are parsimonious in terms of the number of required variational parameters, which is important since a high-dimensional variational optimization is computationally burdensome. The methods suggested improve on Gaussian variational approximation methods for a similar computational cost. Besides defining a novel and useful variational family appropriate to local latent variable models, we also employ importance weighted variational inference methods \citep{Burda2016, Domke2018} to further improve the quality of inference, and elaborate further on the connections between this approach and the use of R\'{e}nyi's divergence within the variational optimization \citep{Li2016, Regli2018, Yang2019}.  

Our method is a contribution to the literature on the development of flexible variational families, and there are many interesting existing methods for this task. One fruitful approach is based on normalizing flows \citep{Rezende2015}, where a variational family is defined using an invertible transformation of a random vector with some known density function. To be useful, the transformation should have an easily computed Jacobian determinant. In the original work of \cite{Rezende2015}, compositions of simple flows called radial and planar flows were considered. Later authors have suggested alternatives, such as autoregressive flows \citep{Germain2015}, inverse autoregressive flows \citep{Kingma2016}, and real-valued non-volume preserving transformations \citep{Dinh2017}, among others. \cite{Spantini2018} gives a theoretical framework connecting Markov properties of a target posterior distribution to representations involving transport maps, with normalizing flows being one way to parametrize such mappings. The variational family we consider here can be thought of as a simple autoregressive flow, but carefully constructed to preserve the conditional independence structure in the true posterior and to achieve parsimony in the representation of dependence between local latent variables and global scale parameters. Our work is also related to the hierarchically structured approximations considered in \citet[Section 7.1]{Salimans2013};  these authors also consider other flexible approximations based on mixture models, and a variety of innovative numerical approaches to the variational optimization. \cite{Hoffman2015} propose an approach called structured stochastic variational inference which is applicable in conditionally conjugate models.  Their approach is similar to ours, in the sense that local variables depend on global variables in the variational posterior.  However conditional conjugacy does not hold in the examples we consider.

The methods we describe can be thought of as extending the Gaussian variational approximation (GVA) of \cite{Tan2018}, where parametrization of the variational covariance matrix was considered in terms of a sparse Cholesky factor of the precision matrix. Similar approximations have been considered for state space models in \cite{Archer2016}.  The sparse structure reduces the number of free variational parameters, and allows matching the exact conditional independence structure in the true posterior. \cite{Tan2018b} propose an approach called reparametrized variational Bayes, where the model is reparametrized by applying an invertible affine transformation to the local variables to minimize their posterior dependency on global variables, before applying a mean field approximation. The  affine transformation is obtained by considering a second order Taylor series approximation to the posterior of the local variables conditional on the global variables. One way of improving on Gaussian approximations is to consider mixtures of Gaussians \citep{Jaakkola1998, Salimans2013, Miller2016, Guo2016}.  However, even with a parsimonious parametrization of component densities, a large number of additional variational parameters are added with each mixture component. Other flexible variational families can be formed using copulas \citep{tran+ba15,Han2016,Smith2019}, hierarchical variational models \citep{Ranganath2016} or implicit approaches \citep{Huszar17}.  

We specify the model and notation in Section \ref{sec_model} and introduce the conditionally structured Gaussian variational approximation (CSGVA) in Section \ref{sec_CSGVA}. The algorithm for optimizing the variational parameters is described in Section \ref{sec_optim} and Section \ref{sec_assoc} highlights the association between GVA and CSGVA. Section \ref{sec_impt} describes how CSGVA can be improved using importance weighting. Experimental results and applications to generalized linear mixed models (GLMMs) and state space models are presented in Sections \ref{sec_expt}, \ref{sec_GLMMs} and \ref{sec_SSMs} respectively. Section \ref{sec_conclu} gives some concluding discussion.

\section{Model specification and notation} \label{sec_model}
Let $y = (y_1, \dots, y_n)^T$ be observations from a model with global variables $\theta_G$ and local variables $\theta_L = (b_1, \dots, b_n)^T$, where $b_i$ contains latent variables specific to $y_i$ for $i=1, \dots, n$. Suppose $\theta_G$ is a vector of length $G$ and each $b_i$ is a vector of length $L$. Let $\theta = (\theta_G^T, \theta_L^T)^T$. We consider models where the joint density is of the form 
\begin{equation*}
\begin{aligned}
p(y, \theta) &= p(\theta_G) p(b_1, \dots, b_\ell|\theta_G)\bigg\{ \prod_{i=1}^n p(y_i|b_i, \theta_G) \bigg\} \\
&\quad  \times \bigg\{  \prod_{i> \ell} p(b_i|b_{i-1}, \dots, b_{i-\ell}, \theta_G) \bigg\}.
\end{aligned}
\end{equation*}
The observations $\{y_i\}$ are conditionally independent given $\{b_i\}$ and $\theta_G$. Conditional on $\theta_G$, the local variables $\{b_i\}$ form a $\ell$th order Markov chain if $\ell \geq 1$, and they are conditionally independent if $\ell=0$. This class of models include important models such as GLMMs and state space models. Next, we define some mathematical notation before discussing CSGVA for this class of models.

\subsection{Notation} \label{sec_notation}
For an $r \times r$ matrix $A$, let $\diag(A)$ denote the diagonal elements of $A$ and $\dg(A)$ be the diagonal matrix obtained by setting non-diagonal elements in $A$ to zero. Let $\text{vec}(A)$ be the vector of length $r^2$ obtained by stacking the columns of $A$ under each other from left to right and $\vech(A)$ be the vector of length $r(r+1)/2$ obtained from $\VEC(A)$ by eliminating all superdiagonal elements of $A$. Let $E_r$ be the $r(r+1)/2 \times r^2$ elimination matrix, $K_r$ be the $r^2 \times r^2$ commutation matrix and $D_r$ be the $r^2 \times r(r+1)/2$ duplication matrix \citep[see][]{Magnus1980}. Then $K_r \VEC(A) = \VEC(A^T)$, $E_r \VEC(A) = \vech(A)$, $E_r^T \vech(A) = \VEC(A)$ if $A$ is lower triangular, and $D_r \vech(A) = \VEC(A)$ if $A$ is symmetric. Let $\bone_r$ be a vector of ones of length $r$. The Kronecker product between any two matrices is denoted by $\otimes$. Scalar functions applied to vector arguments are evaluated element by element. Let $\diff$ denote the differential operator \cite[see e.g.][]{Magnus1999}.

\section{Conditionally structured Gaussian variational approximation} \label{sec_CSGVA}
We propose to approximate the posterior distribution $p(\theta|y)$ of the model defined in Section \ref{sec_model} by a density of the form 
\begin{equation*}
q(\theta) = q(\theta_G) q(\theta_L|\theta_G),
\end{equation*}
where $q(\theta_G) = N(\mu_1, \Omega_1^{-1})$, $q(\theta_L|\theta_G) = N(\mu_2, \Omega_2^{-1})$, and $\Omega_1$ and $\Omega_2$ are the precision (inverse covariance) matrices of $q(\theta_G)$ and $q(\theta_L|\theta_G)$ respectively. Here $\mu_2$ and $\Omega_2$ depend on $\theta_G$, but we do not denote this explicitly for notational conciseness. Let $C_1C_1^T$ and $C_2C_2^T$ be unique Cholesky factorizations of $\Omega_1$ and $\Omega_2$ respectively, where $C_1$ and $C_2$ are lower triangular matrices with positive diagonal entries. We further define $C_1^*$ and $C_2^*$ to be lower triangular matrices of order $G$ and $nL$ respectively such that $C_{r,ii}^* = \log (C_{r,ii})$ and $C_{r,ij}^* = C_{r,ij}$ if $i \neq j$ for $r=1,2$. The purpose of introducing $C_1^*$ and $C_2^*$  is to allow unconstrained optimization of the variational parameters in the stochastic gradient ascent algorithm, since diagonal entries of $C_1$ and $C_2$ are constrained to be positive. Note that $C_2$ and $C_2^*$ also depend on $\theta_G$ but again we do not show this explicitly in our notation.

As $\Omega_2$ depends on $\theta_G$, the joint distribution $q(\theta)$ is generally non-Gaussian even though $q(\theta_G)$ and $q(\theta_L|\theta_G)$ are individually Gaussian. Here we consider a first order approximation and assume that $\mu_2$ and $\vech(C_2^*)$ are linear functions of $\theta_G$:
\begin{equation} \label{linear}
\mu_2 = d + C_2^{-T} D(\mu_1 - \theta_G), \quad  
\vech(C_2^*) = f + F \theta_G.
\end{equation}
In \eqref{linear}, $d$ is a vector of length $nL$, $D$ is a $nL \times G$ matrix, $f$ is a vector of length $nL(nL+1)/2$ and $F$ is a $nL(nL+1)/2 \times G$ matrix. For this specification, $q(\theta)$ is not jointly Gaussian due to dependence of the covariance matrix of $q(\theta_L|\theta_G)$ on $\theta_G$. It is Gaussian if and only if $F \equiv 0$. The set of variational parameters to be optimized is denoted as
\begin{equation*}
\lambda = (\mu_1^T, \vech(C_1^*)^T, d^T, \VEC(D)^T, f^T, \VEC(F)^T)^T.
\end{equation*}

As motivation for the linear approximation in \eqref{linear}, consider the linear mixed model,
\begin{equation*}
\begin{gathered}
y_i = X_i \beta + Z_i b_i + \epsilon_i, \quad (i=1, \dots, n) , \\
\beta \sim N(0, \sigma_\beta^2I_p), \quad b_i \sim N(0, \Lambda), \quad \epsilon_i \sim N(0, \sigma_\epsilon^2 I_{n_i}), \\
\end{gathered}
\end{equation*}
where $y_i$ is a vector of responses of length $n_i$ for the $i$th subject, $X_i$ and $Z_i$ are covariate matrices of dimensions $n_i \times p$ and $n_i \times L$ respectively, $\beta$ is a vector of coefficients of length $p$ and $\{b_i\}$ are random effects. Assume $\sigma_\epsilon^2$ is known. Then the global parameters $\theta_G$ consists of $\beta$ and $\Lambda$. The posterior of $\theta_L$ conditional on $\theta_G$ is 
\begin{equation*}
\begin{aligned}
p(\theta_L|y, \theta_G) &\propto \prod_{i=1}^n p(y_i|\beta, b_i) p(b_i|\Lambda) \\
&\propto   \prod_{i=1}^n \exp [-\{b_i^T ( Z_i^TZ_i /\sigma_\epsilon^2 + \Lambda^{-1}) b_i \\
& \quad - 2 b_i^T  Z_i^T(y_i - X_i \beta) /\sigma_\epsilon^2 \}/2].
\end{aligned}
\end{equation*}
Thus $p(\theta_L|y, \theta_G) = \prod_{i=1}^n p(b_i|y, \theta_G)$, where $p(b_i|y, \theta_G)$ is a normal density with precision matrix $Z_i^TZ_i /\sigma_\epsilon^2+ \Lambda^{-1}$ and mean $( Z_i^TZ_i / \sigma_\epsilon^2 + \Lambda^{-1})^{-1} Z_i^T(y_i - X_i \beta)/\sigma_\epsilon^2$. The precision matrix depends on $\Lambda^{-1}$ linearly and the mean depends on $\beta$ linearly after scaling by the covariance matrix. The linear approximation in \eqref{linear} tries to mimic this dependence relationship.

The proposed variational density is conditionally structured and highly flexible. Such dependence structure is particularly valuable in constructing variational approximations for hierarchical models, where there are global scale parameters in $\theta_G$ which help to determine the scale of local latent variables in the conditional posterior of $\theta_L|\theta_G$. While marginal posteriors of the global variables are often well approximated by Gaussian densities, marginal posteriors of the local variables tend to exhibit more skewness and kurtosis. This deviation from normality can be captured by $q(\theta_L) = \int q(\theta_G) q(\theta_L|\theta_G) d\theta_G$, which is a mixture of normal densities. The formulation in \eqref{linear} also allows for a reduction in the number of variational parameters if conditional independence structure consistent with that in the true posterior is imposed on the variational approximation. 

\subsection{Using conditional independence structure}

\citet{Tan2018} incorporate the conditional independence structure of the true posterior into Gaussian variational approximations by using the fact that zeros in the precision matrix correspond to conditional independence for Gaussian random vectors. This incorporation achieves sparsity in the precision matrix of the approximation and leads to a large reduction in the number of variational parameters to be optimized. For high-dimensional $\theta$, this sparse structure is especially important because a full Gaussian approximation involves learning a covariance matrix where the number of elements grows quadratically with the dimension of $\theta$.  

Recall that $\theta_L=(b_1,\dots, b_n)^T$. Suppose $b_i$ is conditionally independent of $b_j$ in the posterior for $|i-j|>\ell$, given $\theta_G$ and $\{b_j \mid |i-j|\leq \ell\}$. For instance, in a GLMM, $\{b_i\}$ may be subject specific random effects, and these are conditionally independent given the global parameters, so this structure holds with $\ell=0$.  In the case of a state space model for a time series, $\{b_i\}$ are the latent states and this structure holds with $\ell=1$. Note that ordering of the latent variables is important here.  

Now partition the precision matrix $\Omega_2$ of $q(\theta_L|\theta_G)$ into $L \times L$ blocks with $n$ row and $n$ column partitions corresponding to $\theta_L=(b_1,\dots, b_n)^T$. Let $\Omega_{2, ij}$ be the block corresponding to $b_i$ horizontally and $b_j$ vertically for $i,j =1, \dots, n$. If $b_i$ is conditionally independent of $b_j$ for $|i-j|>\ell$, given $\theta_G$ and $\{b_j \mid |i-j|\leq \ell \}$, then we set $\Omega_{2, ij}=0$ for all pairs $(i,j)$ with $|i - j| > \ell$. Let $\mathcal{I}$ denote the indices of elements in $\vech(\Omega_2)$ which are fixed at zero by this conditional independence requirement. If we choose $f_i=0$ and $F_{ij}=0$ for all $i\in \mathcal{I}$ and all $j$ in \eqref{linear}, then $C_2^*$ has the same block sparse structure we desire for the lower triangular part of $\Omega_2$. By Proposition 1 of \cite{Rothman2010}, this means that $\Omega_2$ will have the desired block sparse structure. Hence we impose the constraints $f_i=0$ and $F_{ij}=0$ for $i\in \mathcal{I}$ and all $j$, which reduces the number of variational parameters to be optimized.

\section{Optimization of variational parameters} \label{sec_optim}

To make the dependence on $\lambda$ explicit, write $q(\theta)$ as $q_\lambda(\theta)$. The variational parameters $\lambda$ are optimized by minimizing the Kullback-Leibler divergence between $q_\lambda(\theta)$ and the true posterior $p(\theta|y)$, where 
\begin{multline*}
\KL\{q_\lambda||p(\theta|y)\} = \int q_\lambda(\theta) \log \frac{q_\lambda(\theta)}{p(\theta|y)} d\theta \\
= \log p(y) - \int q_\lambda(\theta) \log \frac{p(y,\theta)}{q_\lambda(\theta)} d\theta \geq 0.
\end{multline*}
Minimizing $\KL\{q_\lambda||p(\theta|y)\}$ is therefore equivalent to maximizing an evidence lower bound $\mL^\VI$ on the log marginal likelihood $\log p(y)$, where
\begin{equation} \label{LB}
\mL^\VI = E_{q_\lambda} \{ \log p(y, \theta) - \log q_\lambda(\theta) \}.
\end{equation}
In \eqref{LB}, $E_{q_\lambda}$ denotes expectation with respect to $q_\lambda(\theta)$.
We seek to maximize $\mL^\VI$ with respect to $\lambda$ using stochastic gradient ascent. Starting with some initial estimate of $\lambda$,  we perform the following update at each iteration $t$,
\begin{equation*}
\lambda_t = \lambda_{t-1} + \rho_t \widehat{\nabla}_\lambda \mL^\VI,
\end{equation*}
where $\rho_t$ represents a small stepsize taken in the direction of the stochastic gradient $\widehat{\nabla}_\lambda \mL^\VI$. The sequence $\{\rho_t\}$ should satisfy the conditions $\sum_t \rho_t = \infty$ and $\sum_t \rho_t^2 < \infty$ \citep{Spall2003}. 

An unbiased estimate of the gradient $\nabla_\lambda \mL^\VI$ can be constructed using \eqref{LB} by simulating $\theta$ from $q_\lambda(\theta)$. However, this approach usually results in large fluctuations in the stochastic gradients. Hence we implement the ``reparametrization trick" \citep{Kingma2014, Rezende2014, Titsias2014}, which helps to reduce the variance of the stochastic gradients. This approach writes $\theta\sim q_\lambda(\theta)$ as a function of the variational parameters $\lambda$ and a random vector $s$ having a density not depending on $\lambda$. To explain further, let $s = [s_1^T, s_2^T]^T$, where $s_1$ and $s_2$ are vectors of length $G$ and $nL$ corresponding to $\theta_G$ and $\theta_L$ respectively. Consider a transformation $\theta = r_\lambda (s)$ of the form 
\begin{equation} \label{reptrick}
\begin{bmatrix} \theta_G \\ \theta_L \end{bmatrix} 
 = \begin{bmatrix}
\mu_1 + C_1^{-T} s_1 \\ 
\mu_2 + C_2^{-T} s_2
\end{bmatrix}.
\end{equation}
Since $\mu_2$ and $C_2$ are functions of $\theta_G$ from \eqref{linear}, 
\begin{equation*}
\begin{aligned}
\mu_2 &= d + C_2^{-T} D(\mu_1 - \theta_G) = d- C_2^{-T} D C_1^{-T} s_1, \\
\vech(C_2^*) &= f + F \theta_G = f + F(\mu_1 + C_1^{-T} s_1).
\end{aligned}
\end{equation*}
Hence $\mu_2$ and $C_2$ are functions of $s_1$, and $\theta_L$ is a function of both $s_1$ and $s_2$. This transformation is invertible since given $\theta$ and $\lambda$, we can first recover $s_1 = C_1^T(\theta_G - \mu_1)$, find $\mu_2$ and $C_2$, and then recover $s_2 = C_2^T(\theta_L - \mu_2)$. Applying this transformation, 
\begin{equation}\label{LB2}
\begin{aligned}
\mL^\VI &= \int \phi(s) \{ \log p(y, \theta) - \log q_\lambda(\theta) \} d s \\
&= E_\phi \{ \log p(y, \theta) - \log q_\lambda(\theta) \},
\end{aligned}
\end{equation}
where $E_\phi$ denotes expectation with respect to $\phi(s)$ and $\theta = r_\lambda(s)$.

\subsection{Stochastic gradients} \label{Sec: stoc_grad}
Next, we differentiate \eqref{LB2} with respect to $\lambda$ to find unbiased estimates of the gradients. As $\log q_\lambda (\theta)$ depends on $\lambda$ directly as well as through $\theta$, applying the chain rule, we have
\begin{align}
\nabla_\lambda\mL^\VI &= E_\phi [ \nabla_\lambda r_\lambda(s) \{\nabla_\theta \log p(y, \theta) -  \nabla_\theta \log q_\lambda(\theta)\} \nonumber \\ 
& \quad -  \nabla_\lambda \log q_\lambda(\theta)  ] \label{TD} \\
&= E_\phi[ \nabla_\lambda r_\lambda(s) \{\nabla_\theta \log p(y, \theta) -  \nabla_\theta \log q_\lambda(\theta)\}  ]. \label{PD}
\end{align}
Note that $E_\phi\{  \nabla_\lambda \log q_\lambda(\theta) \}$ = 0 as it is the expectation of the score function. \cite{Roeder2017} refer to the expressions inside the expectations in \eqref{TD} and \eqref{PD} as the {\em total derivative} and {\em path derivative} respectively. In \eqref{PD}, the contributions to the gradient from $\log p(y, \theta)$ and $\log q_\lambda (\theta)$ cancel each other if $q_\lambda(\theta)$ approximates the true posterior well (at convergence). However, the score function $\nabla_\lambda \log q_\lambda(\theta)$ is not necessarily small even if $q_\lambda(\theta)$ is a good approximation to $p(\theta|y)$. This term affects adversely the ability of the algorithm to converge and ``stick" to the optimal variational parameters, a phenomenon \cite{Roeder2017} refers to as ``sticking the landing". Hence we consider the path derivative,
\begin{equation} \label{grad_est}
\widehat{\nabla}_\lambda \mL^\VI= \nabla_\lambda r_\lambda(s) \{\nabla_\theta \log p(y, \theta) -  \nabla_\theta \log q_\lambda(\theta)\}
\end{equation}
as an unbiased estimate of the true gradient $\nabla_\lambda\mL^\VI$. \cite{Tan2018} and \cite{Tan2018b} also demonstrate that the path derivative has smaller variation about zero when the algorithm is close to convergence. 

Let $\nabla_\theta \log p(y, \theta) -  \nabla_\theta \log q_\lambda(\theta) = (\G_1^T, \G_2^T)^T$, where $\G_1$ and $\G_2$ are vectors of length $G$ and $nL$ respectively corresponding to the partitioning of $\theta=[\theta_G^T, \theta_L^T]^T$. Then $\widehat{\nabla}_\lambda \mL^\VI = \nabla_\lambda r_\lambda (s) (\G_1^T, \G_2^T)^T$ is given by
\begin{equation*}
\begin{bmatrix}
\G_1 +  \nabla_{\mu_1} \theta_L \G_2  \\ 
- D_1^* \vech[ C_1^{-T} s_1 \{\G_1 +  (\nabla_{\mu_1} \theta_L  - D^T C_2^{-1}) \G_2\} ^T C_1^{-T} ]\\ 
\G_2 \\ 
-\VEC( C_2^{-1}\G_2 s_1^T C_1^{-1} )\\ 
\nabla_{f} \theta_L  \G_2\\ 
\VEC(\nabla_{f} \theta_L \G_2\theta_G^T)   \\ 
\end{bmatrix},
\end{equation*}
where 
\begin{equation*}
\begin{aligned}
\nabla_{\mu_1} \theta_L  &= F^T \nabla_{f} \theta_L ,\\
\nabla_{f} \theta_L \G_2 &= - D_2^*\vech \{C_2^{-T} (s_2 - DC_1^{-T} s_1) \G_2^T C_2^{-T} \}.
\end{aligned}
\end{equation*}
Here $D^*_1$ and $D^*_2$ are diagonal matrices of order $G(G+1)/2$ and $nL(nL+1)/2$ respectively such that $\diff \vech(C_r) = D^*_r \diff\vech(C_r^*)$ for $r=1,2$. Formally, $D^*_1 = \diag\{ \vech(\dg(C_1) + \bone_G\bone_G^T - I_G) \}$ and $D^*_2 = \diag\{ \vech(\dg(C_2) + \bone_{nL} \bone_{nL}^T - I_{nL}) \}$. The full expression and derivation of $\nabla_\lambda r_\lambda (s) $ are given in Appendix A. In addition, we show (in Appendix A) that
\begin{equation*}
\begin{aligned}
&\nabla_\theta \log q_\lambda (\theta) 
= \begin{bmatrix}
\nabla_{\theta_G} \log q_\lambda (\theta)  \\ 
\nabla_{\theta_L} \log q_\lambda (\theta)  
\end{bmatrix} \\
&= \begin{bmatrix}
 F^T [\vech(I_{nL})- D_2^* \vech\{(\theta_L - d)s_2^T\}]- C_1 s_1- D^T s_2  \\
-C_2 s_2
\end{bmatrix}.
\end{aligned}
\end{equation*}
$\nabla_\theta \log p(y, \theta)$ is model specific and we discuss the application to GLMMs and state space models in Sections \ref{sec_GLMMs} and \ref{sec_SSMs} respectively.

\subsection{Stochastic variational algorithm}

The stochastic gradient ascent algorithm for CSGVA is outlined in Algorithm 1. For computing the stepsize, we use Adam \citep{Kingma2015}, which uses bias-corrected estimates of the first and second moments of the stochastic gradients to compute adaptive learning rates. 
\begin{Algorithm}[h]
\caption{CSGVA algorithm.}
\centering
\parbox{0.47\textwidth}{
\hrule
\vspace{1mm}
Initialize $\lambda_0=0$, $m_0=0$, $v_0=0$, \\ [1mm]
For $t=1, \dots, N$,
\begin{enumerate}
\item Generate $s \sim N(0, I_{nL+G})$ and compute $\theta = r_{\lambda_{t-1}} (s)$.
\item Compute gradient $g_t = \widehat{\nabla}_\lambda \mL^\VI$.
\item Update biased first moment estimate: \\
$m_t = \tau_1 m_{t-1} + (1-\tau_1) g_t$.
\item Update biased second moment estimate: \\
$v_t = \tau_2 v_{t-1} + (1-\tau_2) g_t^2$.
\item Compute bias-corrected first moment estimate: \\
$\hat{m}_t = m_t/(1-\tau_1^t)$.
\item Compute bias-corrected second moment estimate: \\
$\hat{v}_t = v_t/(1-\tau_2^t)$.
\item Update $\lambda_t = \lambda_{t-1} + \alpha\hat{m}_t/(\sqrt{\hat{v}_t} + \epsilon)$.
\end{enumerate}
\hrule}
\end{Algorithm} 

At iteration $t$, the variational parameter $\lambda$ is updated as $\lambda_t = \lambda_{t-1}+ \Delta_t$. Let $g_t$ denote the stochastic gradient estimate at iteration $t$. In steps 3 and 4, Adam computes estimates of the mean (first moment) and uncentered variance (second moment) of the gradients using exponential moving averages, where $\tau_1,\tau_2 \in [0,1)$ control the decay rates. In step 4, $g_t^2$ is evaluated as $g_t \odot g_t$, where $\odot$ denotes the Hadamard (element-wise) product. As $m_t$ and $v_t$ are initialized as zero, these moment estimates tend to be biased towards zero, especially at the beginning of the algorithm if $\tau_1$, $\tau_2$ are close to one. As $m_t = (1-\tau_1) \sum_{i=1}^t \tau_1^{t-i} g_i$, 
\[
 E(m_t) = E(g_t) (1-\tau_1^t) + \zeta,
\]
where $\zeta=0$ if $E(g_i) = E(g_t)$ for $1 \leq i < t$. Otherwise, $\zeta$ can be kept small since the weights for past gradients decrease exponentially. An analogous argument holds for $v_t$. Thus the bias can be corrected by using the estimates $\hat{m}_t$ and $\hat{v}_t$ in steps 5 and 6. The change is then computed as
\begin{equation*}
\Delta_t = \frac{\alpha\hat{m}_t}{\sqrt{\hat{v}_t} + \epsilon},
\end{equation*}
where $\alpha$ controls the magnitude of the stepsize and $\epsilon$ is a small positive constant which ensures that the denominator is positive. In our experiments, we set $\alpha = 0.001$, $\tau_1 = 0.9$, $\tau_2=0.99$ and $\epsilon = 10^{-8}$, values close to what is recommended by \cite{Kingma2015}.

At each iteration $t$, we can also compute an unbiased estimate of the lower bound,  
\[
\hat{\mL }^\VI= \log p(y, \theta) - \log q_{\lambda_{t-1}}(\theta),
\]
where $\theta$ is computed in step 1. Since these estimates are stochastic, we follow the path traced by $\bar{\mL}^\VI $, which is an average of the lower bounds averaged over every 1000 iterations, as a means to diagnose the convergence of Algorithm 1. $\bar{\mL^\VI }$ tends to increase monotonically at the start, but as the algorithm comes close to convergence, the values of $\bar{\mL}^\VI $ fluctuate close to and about the true maximum lower bound. Hence, we fit a least squares regression line to the past $\kappa$ values of $\bar{\mL}^\VI $ and terminate Algorithm 1 once the gradient of the regression line becomes negative \citep[see][]{Tan2018b}. For our experiments, we set $\kappa = 6$.

\section{Links to Gaussian variational approximation} \label{sec_assoc}
CSGVA is an extension of Gaussian variational approximation \citep[GVA,][]{Tan2018}. In both approaches, the conditional posterior independence structure of the local latent variables is used to introduce sparsity in the precision matrix of the approximation. Below we demonstrate that GVA is a special case of CSGVA when $F \equiv 0$. 

\cite{Tan2018} consider a GVA of the form
\begin{equation*}
\begin{bmatrix} \theta_L \\ \theta_G \end{bmatrix} \sim
\N\left( \begin{bmatrix} \mu_L \\ \mu_G \end{bmatrix}, T^{-T} T^{-1}
\right) \,\text{where} \; 
T = \begin{bmatrix} T_{LL} & 0 \\
T_{GL} & T_{GG} \end{bmatrix}.
\end{equation*}
Note that $T_{LL}$ and $T_{GG}$ are lower triangular matrices. Using a vector $s = [s_{L}^T, s_G^T]^T \sim N(0, I_{nL+G})$, we can write 
\begin{equation*}
\begin{aligned}
&\begin{bmatrix} \theta_L \\ \theta_G \end{bmatrix} 
= \begin{bmatrix} \mu_L \\ \mu_G \end{bmatrix} + T^{-T} 
\begin{bmatrix} s_{L} \\s_G \end{bmatrix}, \\
&\text{where} \quad
T^{-T} = \begin{bmatrix} T_{LL}^{-T}  & - T_{LL}^{-T} T_{GL}^T T_{GG}^{-T} \\[2mm]
0 & T_{GG}^{-T} \end{bmatrix}.
\end{aligned}
\end{equation*}
Assuming $F\equiv 0$ for CSGVA, we have from \eqref{reptrick} that 
\begin{equation*}
\begin{aligned}
\begin{bmatrix} \theta_L \\ \theta_G \end{bmatrix} 
= \begin{bmatrix} d \\ \mu_1 \end{bmatrix} + 
\begin{bmatrix}
C_2^{-T} & - C_2^{-T} D C_1^{-T} \\
0 & C_1^{-T} & 
\end{bmatrix}
\begin{bmatrix} s_2 \\s_1 \end{bmatrix},
\end{aligned}
\end{equation*}
where $[s_2^T, s_1^T]^T \sim N(0, I_{nL+G})$. Hence we can identify 
\begin{equation*}
\mu_1 = \mu_G, \; d = \mu_L, \; C_1 = T_{GG}, \; C_2 = T_{LL}, \; D = T_{GL}^T.
\end{equation*}
If the standard way of initializing of Algorithm 1 (by setting $\lambda=0$) does not work well, we can use this association to initialize Algorithm 1 by using the fit from GVA. This informative initialization can reduce computation time significantly although there may be a risk of getting stuck in a local mode.

\section{Importance weighted variational inference} \label{sec_impt}
Here we discuss how CSGVA can be improved by maximizing an importance weighted lower bound \citep[IWLB,][]{Burda2016}, which leads to a tighter lower bound on the log marginal likelihood, and a variational approximation less prone to underestimation of the true posterior variance. We also relate the IWLB with R\'{e}nyi's divergence \citep{Renyi1961, vanErven2014} between $q_\lambda(\theta)$ and $p(\theta|y)$, demonstrating that maximizing the IWLB instead of the usual evidence lower bound leads to a transition in the behavior of the variational approximation from ``mode-seeking" to ``mass-covering". We first define R\'{e}nyi's divergence and the variational R\'{e}nyi bound \citep{Li2016}, before introducing the IWLB as the expectation of a Monte Carlo approximation of the variational R\'{e}nyi  bound. 

\subsection{R\'{e}nyi's divergence and variational R\'{e}nyi bound}
R\'{e}nyi's divergence provides a measure of the distance between two densities $q$ and $p$, and it is defined as 
\begin{equation*}
D_\alpha(q||p) = \frac{1}{\alpha-1} \log \int q(\theta)^\alpha p(\theta)^{1-\alpha} d\theta,
\end{equation*}
for $0 <\alpha <\infty$, $\alpha \neq 1$. This definition can be extended by continuity to the orders 0, 1 and $\infty$, as well as to negative orders $\alpha < 0$. Note that $D_\alpha(q||p)$ is no longer a divergence measure if $\alpha < 0$, but we can write $D_\alpha(q||p)$ as $\frac{\alpha}{1-\alpha} D_{1-\alpha}(p||q)$ for $\alpha \notin \{0,1\}$ by the skew symmetry property. As $\alpha$ approaches 1, the limit of $D_\alpha(q||p)$ is the Kullback-Leibler divergence, $\KL(q||p)$. In variational inference, minimizing the Kullback-Leibler divergence between the variational density $q_\lambda(\theta)$ and the true posterior $p(\theta|y)$ is equivalent to maximizing a lower bound $\mL^{\VI}$ on the log marginal likelihood due to the relationship:
\begin{equation*}
\begin{aligned}
\mL^\VI &= \log p(y) - \KL\{q_\lambda||p(\theta|y)\} = E_{q_\lambda}\left\{ \log \frac{ p(y, \theta)}{q_\lambda(\theta)} \right\} .
\end{aligned}
\end{equation*}
Generalizing this relation using R\'{e}nyi's divergence measure, \cite{Li2016} define the variational R\'{e}nyi bound $\mL_\alpha$ as 
\begin{equation*}
\begin{aligned}
\mL_\alpha &= \log p(y) - D_\alpha\{q_\lambda||p(\theta|y)\} \\
&= \frac{1}{1-\alpha} \log E_{q_\lambda}\left\{ \left(\frac{ p(y, \theta)  }{q_\lambda(\theta)} \right)^{1-\alpha} \right\} .
\end{aligned}
\end{equation*}
Note that $\mL_1$, the limit of $\mL_\alpha$ as $\alpha \rightarrow 1$, is equal to $\mL^\VI$. A Monte Carlo approximation of $\mL_\alpha$ when the expectation is intractable is  
\begin{equation} \label{Lhatalpha}
\hat{\mL}_{\alpha, K} = \frac{1}{1-\alpha} \log \frac{1}{K} \sum_{k=1}^K w_k^{1-\alpha}, 
\end{equation}
where $\Theta_K = [\theta_1, \dots, \theta_K]$ is a set of $K$ samples generated randomly from $q_\lambda(\theta)$, and 
\[
w_k = w(\theta_k) = \frac{p(y, \theta_k)}{q_\lambda(\theta_k)}, \quad k=1, \dots, K,
\]
are importance weights. For each $k$, $E_{q_\lambda} (w_k) = p(y)$. The limit of $\hat{\mL}_{\alpha, K}$ as $\alpha \rightarrow 1$ is $\frac{1}{K} \sum_{k=1}^K \log \frac{p(y, \theta_k)}{q_\lambda(\theta_k)}$. Hence $\hat{\mL}_{1, K}$ is an unbiased estimate of $\mL_1$ as $E_{\Theta_K}(\hat{\mL}_{1, K}) = \mL_1 = \mL^\VI$, 
where $E_{\Theta_K}$ denotes expectation with respect to $q(\Theta_K) = \prod_{k=1}^{K} q_\lambda (\theta_k)$. For $\alpha \neq 1$, $\hat{\mL}_{\alpha, K} $ is not an unbiased estimate of $\mL_\alpha$.

\subsection{Importance weighted lower bound}
The importance weighted lower bound \citep[IWLB,][]{Burda2016} is defined as
\begin{equation*}
\mL_K^\IW = E_{\Theta_K}(\hat{\mL}_{0, K}) 
= E_{\Theta_K}\left( \log  \frac{1}{K} \sum_{k=1}^K w_k  \right), 
\end{equation*}
where $\alpha = 0$ in \eqref{Lhatalpha}. It reduces to $\mL^{\VI}$ when $K=1$. By Jensen's inequality,
\begin{equation*}
\begin{aligned}
\mL_K^\IW \leq  \log E_{\Theta_K} \left( \frac{1}{K} \sum_{k=1}^K w_k \right) = \log p(y).
\end{aligned}
\end{equation*}
Thus $\mL_K^\IW$ provides a lower bound to the log marginal likelihood for any positive integer $K$. From Theorem \ref{thm1}  \citep{Burda2016}, this bound becomes tighter as $K$ increases.
\begin{theorem} \label{thm1}
$\mL^\IW_K$ increases with $K$ and approaches $\log p(y)$ as $K \rightarrow \infty$ if $w_k$ is bounded. 
\end{theorem}
\begin{proof}
Let $I = \{w_{I_1}, \dots, w_{I_K}\}$ be selected randomly without replacement from $\{w_1, \dots, w_{K+1}\}$. Then $E_{I|\Theta_{K+1}}(w_{I_j}) = \frac{1}{K+1} \sum_{k=1}^{K+1} w_k$ for any $j=1, \dots, K$, where $E_{I|\Theta_{K+1}}$ denotes the expectation associated with the randomness in selecting $I$ given $\Theta_{K+1}$. Thus 
\begin{equation*}
\begin{aligned}
\mL^\IW_{K+1} &= E_{\Theta_{K+1}} \left( \log \frac{1}{K+1} \sum\nolimits_{k=1}^{K+1} w_k \right) \\
&= E_{\Theta_{K+1}} \left\{ \log E_{I|\Theta_{K+1}} \left( \frac{1}{K} \sum\nolimits_{j=1}^K w_{I_j} \right) \right\}  \\
&\geq E_{\Theta_{K+1}} \left\{ E_{I|\Theta_{K+1}} \log\left( \frac{1}{K} \sum\nolimits_{j=1}^K w_{I_j}\right) \right\} \\
&= E_{\Theta_{K}} \log\left( \frac{1}{K} \sum\nolimits_{k=1}^K w_k \right) = \mL^\IW_K.
\end{aligned}
\end{equation*}
If $w_k$ is bounded, then $\frac{1}{K} \sum_{k=1}^K w_k \overset{a.s.}{\rightarrow} p(y)$ as $K \rightarrow \infty$ by the law of large numbers. Hence $\mL^\IW_K = E_{\Theta_K} ( \log \frac{1}{K} \sum_{k=1}^K w_k ) \rightarrow \log p(y)$ as $K \rightarrow \infty$.
\end{proof}

Next, we present some properties of R\'{e}nyi's divergence and $E_{\Theta_K}(\hat{\mL}_{\alpha, K})$ which are important in understanding the behavior of the variational density arising from maximizing $\mL^\IW_K$. The proofs of these properties can be found in \cite{vanErven2014} and \cite{Li2016}. 

\begin{property}
$D_\alpha$ is increasing in $\alpha$, and is continuous in $\alpha$ on $[0,1] \cup \{\alpha\notin[0,1] \mid |D_\alpha| < \infty \}$. 
\end{property}

\begin{property}
$E_{\Theta_K}(\hat{\mL}_{\alpha, K})$ is continuous and decreasing in $\alpha$ for fixed $K$. 
\end{property}

\begin{theorem} \label{thm2}
There exists $0 \leq \alpha_{q_\lambda,K} \leq 1$ for given $q_\lambda$ and $K$ such that 
\[
\log p(y) = D_{\alpha_{q_\lambda,K}}\{q_\lambda||p(\theta|y)\} + \mL_K^\IW.
\]
\end{theorem} 
\begin{proof}
From Property 2,
\begin{equation*}
\begin{gathered}
\mL_1 =E_{\Theta_K}(\hat{\mL}_{1, K})  \leq E_{\Theta_K}(\hat{\mL}_{0, K}) = \mL^\IW_K \leq  \mL_0, \\
\mL_1- \log p(y) \leq \mL^\IW_K - \log p(y) \leq \mL_0 -  \log p(y), \\
D_0\{ q_\lambda ||p(\theta|y)\} \leq \log p(y) - \mL^\IW_K  \leq D_1\{ q_\lambda ||p(\theta|y)\} .
\end{gathered}
\end{equation*}
From Property 1, since $D_\alpha$ is continuous and decreasing for $\alpha \in [0,1]$, there exists $0 \leq \alpha_{q_\lambda,K} \leq 1$ such that $\log p(y) - \mL^\IW_K = D_{\alpha_{q_\lambda,K}}\{q_\lambda||p(\theta|y)\}$.
\end{proof}
Minimizing R\'{e}nyi's divergence for $\alpha \gg 1$ tends to produce approximations which are mode-seeking (zero-forcing) while maximizing R\'{e}nyi's divergence for $\alpha \ll 0$ encourages mass-covering behavior. Theorem \ref{thm2} suggests that maximizing the IWLB results in a variational approximation $q_\lambda$ whose R\'{e}nyi's divergence from the true posterior can be captured with $0 \leq \alpha \leq 1$, which represents a mix and certain balance between mode-seeking and mass-covering behavior \citep{Minka2005}. In our experiments, we observe that maximizing the IWLB is highly effective in correcting the underestimation of posterior variance in variational inference. 

Alternatively, if we approximate $\mL_K^\IW$ by considering a second-order Taylor expansion of $\log \bar{w}_K$ about $E_\Theta(\bar{w}_K) = p(y)$, where $\bar{w}_K =\frac{1}{K} \sum_{k=1}^K w_K$, and then take expectations, we have 
\[
\mL_K^\IW \approx \log p(y) - \frac{\var(w_k)}{2K p(y)^2}.
\]
\cite{Maddison2017} and \cite{Domke2018} provide bounds on the order of the remainder term in the Taylor approximation above, and demonstrate that the ``looseness" of the IWLB is given by  $\var(w_k) $ as $K \rightarrow \infty$. Minimizing $\var(w_K)$ is equivalent to minimizing the $\chi^2$ divergence $D_2(p||q)$. Note that if $q_\lambda(\theta)$ has thin tails compared to $p(\theta|y)$, then the variance of $\var(w_k)$ will be large. Hence minimizing $\var(w_K)$ attempts to match $p(\theta|y)$ with $q_\lambda(\theta)$ so that $q_\lambda(\theta)$ is able to cover the tails.

\subsection{Unbiased gradient estimate of importance weighted lower bound}
To maximize the IWLB in CSGVA, we need to find an unbiased estimate of $\nabla_\lambda \mL_K^\IW $ using the transformation in \eqref{reptrick}. Let $s_k \sim N(0, I_{G+nL})$, $\theta_k = r _\lambda(s_k)$ for $k=1, \dots, K$, and $S_K = [s_1, \dots, s_K]^T$.
\begin{equation} \label{IW1}
\begin{aligned}
\nabla_\lambda \mL_K^\IW
&= \nabla_\lambda E_{\Theta_K} (\log \bar{w}_K )= \nabla_\lambda E_{S_K} ( \log \bar{w}_K )\\
&= E_{S_K} \left[\sum_{k=1}^K \frac{ \nabla_\lambda w_k }{ \sum_{k'=1}^K w_{k'}} \right] \\
&= E_{S_K} \left[\sum_{k=1}^K \frac{ w_k \nabla_\lambda \log w_k }{ \sum_{k'=1}^K w_{k'} } \right] \\
&= E_{S_K} \left[\sum\nolimits_{k=1}^K \widetilde{w}_k \nabla_\lambda \log w_k \right],
\end{aligned}
\end{equation}
where $w_k = w(\theta_k) = w\{r _\lambda(s_k)\}$ and $\widetilde{w}_k = w_k/ \sum_{k'=1}^K w_{k'}$ for $k=1, \dots, K$ are normalized importance weights. Applying chain rule,
\begin{equation} \label{IW2}
\begin{aligned}
\nabla_\lambda \log w_k &= \nabla_\lambda r_\lambda(s_k) \nabla_{\theta_k} \log w_k - \nabla_\lambda \log q_\lambda(\theta_k).
\end{aligned}
\end{equation}
In Section \ref{Sec: stoc_grad}, we note that $E_\phi\{ \nabla_\lambda \log q_\lambda(\theta) \} = 0$ as it is the expectation of the score function and hence we can omit $\nabla_\lambda \log q_\lambda(\theta)$ to obtain an unbiased estimate of $\nabla_\lambda \mL^\VI$. However, in this case, it is unclear if 
\begin{equation} \label{grad_lambda_q}
E_{S_K} \left[\sum_{k=1}^K \widetilde{w}_k \nabla_\lambda \log q_\lambda(\theta_k) \right] = 0.
\end{equation}
\cite{Roeder2017} conjecture that \eqref{grad_lambda_q} is true and report improved results when omitting the term $\nabla_\lambda \log q_\lambda(\theta_k)$ from $\nabla_\lambda \log w_k$ in computing gradient estimates. However, \cite{Tucker2018} demonstrated via simulations that \eqref{grad_lambda_q} does not hold generally and that such omission will result in biased gradient estimates. Our own simulations using CSGVA also suggest that \eqref{grad_lambda_q} does not hold even though omission of $\nabla_\lambda \log q_\lambda(\theta_k)$ does lead to improved results. As the stochastic gradient algorithm is not guaranteed to converge with biased gradient estimates, we turn to the {\em double reparametrized gradient estimate} proposed by \cite{Tucker2018} which allows unbiased gradient estimates to be constructed with the omission of $\nabla_\lambda \log q_\lambda(\theta_k)$ albeit with revised weights. 

Since $\tilde{w}_k$ depends on $\lambda$ directly as well as through $\theta_k$, we use chain rule to obtain
\begin{equation} \label{LHS}
\begin{aligned}
\nabla_\lambda E_{\theta_k}(\tilde{w}_k) &= \nabla_\lambda E_{s_k}(\tilde{w}_k) \\
&= E_{s_k}(\nabla_\lambda \theta_k \nabla_{\theta_k} \tilde{w}_k) + E_{s_k}(\nabla_\lambda \tilde{w}_k),
\end{aligned}
\end{equation}
where 
\begin{equation*}
\begin{aligned}
\nabla_{\theta_k} \tilde{w}_k &= \left\{  \frac{1}{\sum_{k'=1}^K w_{k'}} - \frac{w_k }{(\sum_{k'=1}^K w_{k'})^2} \right\} \nabla_{\theta_k} w_k \\
&= ( \tilde{w}_k- \tilde{w}_k^2 ) \nabla_{\theta_k} \log w_k.
\end{aligned}
\end{equation*}
Alternatively,
\begin{equation} \label{RHS}
\begin{aligned}
&\nabla_\lambda E_{\theta_k}(\tilde{w}_k)
=\nabla_\lambda \int q_\lambda(\theta_k) \tilde{w}_k  \;d\theta_k \\
&= \int \tilde{w}_k \nabla_\lambda q_\lambda(\theta_k) + q_\lambda(\theta_k) \nabla_\lambda \tilde{w}_k \; d\theta_k \\
&= \int \tilde{w}_k q_\lambda(\theta_k)  \nabla_\lambda \log q_\lambda(\theta_k) \;d\theta_k  + E_{\theta_k}(\nabla_\lambda \tilde{w}_k) \\
&= E_{\theta_k} [\tilde{w}_k \nabla_\lambda \log q_\lambda(\theta_k) ] + E_{s_k} (\nabla_\lambda \tilde{w}_k).
\end{aligned}
\end{equation}
Comparing \eqref{LHS} and \eqref{RHS}, we have 
\begin{equation*}
\begin{aligned}
&E_{\Theta_K} \left( \sum\nolimits_{k=1}^K \tilde{w}_k \nabla_\lambda \log q_\lambda(\theta_k) \right) \\ 
&= \sum\nolimits_{k=1}^K E_{\Theta_K\backslash \theta_k} E_{\theta_k}  [\tilde{w}_k \nabla_\lambda \log q_\lambda(\theta_k)]  \\
&= \sum\nolimits_{k=1}^K E_{S_K\backslash s_k} E_{s_k}(\nabla_\lambda \theta_k ( \tilde{w}_k- \tilde{w}_k^2 ) \nabla_{\theta_k} \log w_k)  \\
&= E_{S_K} \left\{ \sum\nolimits_{k=1}^K ( \tilde{w}_k  - \tilde{w}_k^2 ) \nabla_\lambda r_\lambda(s_k) \nabla_{\theta_k} \log w_k \right\}.
\end{aligned}
\end{equation*}
Combining the above expression with \eqref{IW1} and \eqref{IW2}, we find that 
\begin{equation*}
\nabla_\lambda \mL_K^\IW = E_{S_K} \left( \sum\nolimits_{k=1}^K \tilde{w}_k^2 \nabla_\lambda r_\lambda(s_k) \nabla_{\theta_k} \log w_k \right)
\end{equation*}
An unbiased gradient estimate is thus given by 
\begin{equation*}
\widehat{\nabla}_\lambda \mL_K^\IW = \sum_{k=1}^K \tilde{w}_k^2 \nabla_\lambda r_\lambda(s_k) \nabla_{\theta_k} \{\log p(y, \theta_k) - \log q_\lambda (\theta_k) \}.
\end{equation*}

Thus, to use CSGVA with important weights, we only need to modify Algorithm 1 by drawing $K$ samples $s_1, \dots, s_K$ in step 1 instead of a single sample and then compute the unbiased gradient estimate, $g_t = \widehat{\nabla}_\lambda \mL_K^\IW $, in step 2. The rest of the steps in Algorithm 1 remain unchanged. In the importance weighted version of CSGVA, the gradient estimate based on a single sample $s$ is replaced by a weighted sum of the gradients in \eqref{grad_est} based on $K$ samples $s_1, \dots, s_K$. However, these weights do not necessarily sum to 1. An unbiased estimate of $\mL_K^\IW$ itself is given by $\hat{\mL}_K^\IW = \log \frac{1}{K} \sum_{k=1}^K w_k$.

\section{Experimental results} \label{sec_expt}
We apply CSGVA to GLMMs and state space models and compare the results with that of GVA in terms of computation time and accuracy of the posterior approximation. Lower bounds reported exclude constants which are independent of the model variables. We also illustrate how CSGVA can be improved using importance weighting (IW-CSGVA), considering $K \in \{5, 20, 100\}$. We find that IW-CSGVA performs poorly if it is initialized in the standard manner using $\lambda = 0$. This is because, when $q_\lambda(\theta)$ is still far from optimal, a few of the importance weights tend to dominate with the rest close to zero, thus producing inaccurate estimates of the gradients. Hence we initialize IW-CSGVA using the CSGVA fit, and the algorithm is terminated after a short run of 1000 iterations as IW-CSGVA is computationally intensive and improvements in the IWLB and variational approximation seem to be negligible thereafter. Posterior distributions estimated using MCMC via RStan are regarded as the ground truth. Code for all variational algorithms are written in Julia and are available as supplementary materials. All experiments are run on on Intel Core i9-9900K CPU @3.60 GHz with 16GB RAM. As the computation time of IW-CSGVA increases linearly with $K$, we also investigate the speedup that can be achieved through parallel computing on a machine with 8 cores. Julia retains one worker (or core) as the master process and parallel computing is implemented using the remaining seven workers. 

The parametrization of a hierarchical model plays a major role in the rate of convergence of both GVA and CSGVA. In some cases, it can even affect the ability of the algorithm to converge (to a local mode). We have attempted both the centered and noncentered parametrizations \citep{Papaspiliopoulos2003, Papaspiliopoulos2007}, which are known to have a huge impact on the rate of convergence of MCMC algorithms. These two parametrizations are complementary and neither is superior to the other. If an algorithm converges slowly under one parametrization, it often converges much faster under the other. Which parametrization works better usually depends on the nature of data. For the datasets that we use in the experiments, the centered parametrization was found to have better convergence properties than the noncentered parametrization for GLMMs while the noncentered parametrization is preferred for stochastic volatilty models. These observations are further discussed below.

 \section{Generalized linear mixed models} \label{sec_GLMMs}
Let $y_i = (y_{i1}, \dots, y_{in_i})^T$ denote the vector of responses of length $n_i$ for the $i$th subject for $ i=1,\dots,n$, where $y_{ij}$ is generated from some distribution in the exponential family. The mean $\mu_{ij} = E(y_{ij})$ is connected to the linear predictor $\eta_{ij}$ via
\begin{equation*}
g(\mu_{ij}) = \eta_{ij} = X_{ij}^T\beta + Z_{ij}^Tb_i,
\end{equation*}
for some smooth invertible link function $g(\cdot)$. The fixed effects $\beta$ is a $p \times 1$ vector and $b_i \sim N(0,\Lambda)$ is a $L \times 1$ vector of random effects specific to the $i$th subject. $X_{ij}$ and $Z_{ij}$ are  vectors of covariates of length $p$ and $L$ respectively. Let $\eta_i = [\eta_{i1}, \dots, \eta_{in_i}]^T$, $X_i = [X_{i1}, \dots, X_{in_i}]^T$ and $Z_i = [Z_{i1}, \dots, Z_{in_i}]^T$. We focus on the one-parameter exponential family with canonical links. This includes the Bernoulli model, $y_{ij} \sim$ Bern($\mu_{ij}$), with the logit link $g(\mu_{ij}) = \log\{\mu_{ij}/(1-\mu_{ij})\}$ and the Poisson model, $y_{ij} \sim $ Pois($\mu_{ij}$), with the log link $g(\mu_{ij}) = \log(\mu_{ij})$. Let $WW^T$ be the unique Cholesky decomposition of $\Lambda^{-1}$, where $W$ is a lower triangular matrix with positive diagonal entries. Define $W^*$ such that $W^*_{ii} = \log(W_{ii})$ and $W^*_{ij} = W_{ij}$ if $i \neq j$, and let $\omega = \vech(W^*)$.  We consider normal priors, $\beta \sim N(0, \sigma_{\beta}^2I)$ and $\omega \sim N(0, \sigma_\omega^2 I)$, where $\sigma_\beta^2$ and $\sigma_\omega^2$ are set as 100. 

The above parametrization of the GLMM is {\em noncentered} since $b_i$ has mean 0. Alternatively, we can consider the {\em centered} parametrization proposed by \cite{Tan2013}. Suppose the covariates for the random effects are a subset of the covariates for the fixed effects and the first column of $X_i$ and $Z_i$ are ones corresponding to an intercept and random intercept respectively. Then we can write 
\[
\eta_i = X_i \beta + Z_i b_i = Z_i \beta_R + X_i^G \beta_G + Z_i b_i.
\]
where $\beta = [\beta_R^T, \beta_G^T]^T$. We further split $X_i^G$ into covariates which are subject specific (varies only with $i$ and assumes the same value for $j=1, \dots, n_i$) and those which are not, and $\beta_G = [ \beta_{G_1}^T, \beta_{G_2}^T]^T$ accordingly, where $\beta_{G_1}$, $\beta_{G_2}$ are vectors of length $g_1$ and $g_2$ respectively. Then 
\begin{equation*}
\begin{aligned}
\eta_i &= Z_i \beta_R + \bone_{n_i} {x_i^{G_1}}^T \beta_{G_1} + X_i^{G_2} \beta_{G_2} + Z_i b_i \\
& = Z_i (C_i \beta_{RG_1} + b_i) + X_i^{G_2} \beta_{G_2},
\end{aligned}
\end{equation*}
where
\begin{equation*}
C_i =\begin{bmatrix}  I_L & \begin{array}{l} {x_i^{G_1}}^T \\ 0_{L-1 \times g_1} \end{array}  \end{bmatrix}\;\; \text{and}\;\;\beta_{RG_1}= \begin{bmatrix} \beta_R \\ \beta_{G_1} \end{bmatrix}.
\end{equation*}
Let $\tilde{b}_i = C_i \beta_{RG_1}+ b_i$. The centered parametrization is represented as
\begin{equation} \label{centered_GLMM}
\eta_i = Z_i \tilde{b}_i  + X_i^{G_2} \beta_{G_2}, \quad \tilde{b}_i \sim N(C_i \beta_{RG_1}, \Lambda)
\end{equation}
for $i=1, \dots, n$.

\cite{Tan2013} compare the centered, noncentered and partially noncentered parametrizations for GLMMs in the context of variational Bayes, showing that the choice of parametrization affects not only the rate of convergence, but also the accuracy of the variational approximation. For CSGVA, we observe that the accuracy of the variational approximation and the rate of convergence can also be greatly affected. \cite{Tan2013} demonstrate that the centered parametrization is preferred when the observations are highly informative about the latent variables. In practice, a general guideline is to use the centered parametrization for Poisson models when observed counts are large and the noncentered parametrization when most counts are close to zero. For Bernoulli models, differences between using centered and noncentered parametrizations are usually minor. Here we use the centered parametrization in \eqref{centered_GLMM}, which has been observed to yield gains in convergence rates for the datasets used for illustration. 

The global parameters are $\theta_G = (\beta^T, \omega^T)^T$ of dimension $G = p + L(L+1)/2$, and the local variables are $\theta_L = (\tilde{b}_1, \dots, \tilde{b}_n)^T$. The joint density is
\begin{equation*}
p(y,\theta) = p(\beta) p(\omega) \prod^n_{i=1} \bigg\{ p(\tilde{b}_i|\omega, \beta)\prod^{n_i}_{j=1}p(y_{ij}|\beta,\tilde{b}_i) \bigg\}.
\end{equation*}
The log of the joint density is given by
\begin{equation*}
\begin{aligned}
\log p(y, \theta) &= \sum\nolimits_{i=1}^n \{ y_i \eta_i - \bone^T h(\eta_i) \\
& - {(\tilde{b}_i - C_i \beta_{RG_1})^T WW^T (\tilde{b}_i - C_i \beta_{RG_1})}/{2} \}   \\
& - \beta^T \beta / (2\sigma_\beta^2) - \omega^T \omega/ (2\sigma_\omega^2) + n \log|W| + c,
\end{aligned}
\end{equation*}
where $h(\cdot)$ is the log-partition function and $c$ is a constant independent of $\theta$. For the Poisson model with log link, $h(x) = \exp(x)$. For the Bernoulli model with logit link, $h(x) = \log\{ 1+\exp(x) \}$. The gradient $\nabla_\theta \log p(y, \theta) $ is given in Appendix B.

For the GLMM, $b_i$ and $b_j$ are conditionally independent  given $\theta_G$ for $ i \neq j$ in $p(\theta|y)$. Hence we impose the following sparsity structure on $\Omega_2$ and $C_2$,
\begin{equation*}
\begin{aligned}
\Omega_2 = \begin{bmatrix}
\Omega_{2,11} & 0 & \dots & 0	\\
0 & \Omega_{2,22} & \dots & 0	\\
\vdots & \vdots & \ddots & \vdots \\
0 & 0 & \dots & \Omega_{2, nn}
\end{bmatrix}, \\
C_2^* = \left[\begin{array}{cccc}
C^*_{2,11} & 0 & \dots & 0	\\
0 & C^*_{2, 22} & \dots & 0	\\
\vdots & \vdots & \ddots & \vdots \\
0 & 0 & \dots & C^*_{2, nn}
\end{array}\right],
\end{aligned}
\end{equation*}
where each $\Omega_{2,ii}$ is a $L \times L$ block matrix and each $C^*_{2,ii}$ is a $L \times L$ lower triangular matrix. We set $f_i = 0$ and $F_{ij} =0$ for all $i \in \mathcal{I}$ and all $j$, where $\mathcal{I}$ denotes the set of indices in $\vech(C_2^*)$ which are fixed as zero. The number of nonzero elements in $\vech(C_2^*)$ is $nL(L+1)/2$. Hence the number of variational parameters to be optimized are reduced from $nL(nL+1)/2$ to $nL(L+1)/2$ for $f$ and from $nL(nL+1)G/2$ to $nL(L+1)G/2$ for $F$.

\subsection{Epilepsy data}
In this epilepsy data \citep{Thall1990}, $n = 59$ patients are involved in a clinical trial to investigate the effects of the anti-epileptic drug Progabide. The patients are randomly assigned to receive either the drug (Trt = 1) or a placebo (Trt = 0). The response $y_i$ denotes the number of epileptic attacks patient $i$ had during 4 follow-up periods of two weeks each. Covariates include the log of the age of the patients (Age), the log of 1/4 the baseline seizure count recorded over an eight-week period prior to the trial (Base) and Visit (coded as Visit$_1 = -0.3$, Visit$_2 = -0.1$, Visit$_3  = 0.1$ and Visit$_4 = 0.3$). Note that Age is centered about its mean. Consider $y_{ij} \sim \text{Pois}(\mu_{ij})$, where
\begin{equation*}
\begin{aligned}
\log \mu_{ij} &= \beta_0 + \beta_{\text{Base}}\text{Base}_i +\beta_{\text{Trt}}\text{Trt}_i +\beta_{\text{Age}}\text{Age}_i\\
& \quad  + \beta_{\text{Base}\times \text{Trt}}\text{Base}_i \times \text{Trt}_i + \beta_{\text{Visit}}\text{Visit}_{ij} \\
& \quad + b_{i1} + b_{i2}\text{Visit}_{ij},
\end{aligned}
\end{equation*}
for $i = 1,\dots, 59, j = 1,\dots, 4$ \citep{Breslow1993}.

Table \ref{tab_epil} shows the results obtained from applying the variational algorithms to this data. The lower bounds are estimated using 1000 simulations in each case and the mean and standard deviation from these simulations are reported. CSGVA produced an improvement in the estimate of the lower bound (3139.2) as compared to GVA (3138.3) and maximizing the IWLB led to further improvements. Using a larger $K$ of 20 or 100 resulted in minimal improvements compared with $K=5$. As this dataset is small, parallel computing is slower than serial for a small $K$. This is because, even though the importance weights and gradients for $K$ samples are computed in parallel, the cost of sending and fetching data from the workers at each iteration dwarfs the cost of computation when $K$ is small. For $K=100$, parallel computing reduces the computation time by about half.
\begin{table}[htb!]
\caption{Epilepsy data. Number of iterations $I$ (in thousands), runtimes (in seconds) and estimates of lower bound (standard deviation in brackets) of the variational methods.}
\label{tab_epil}
\centering 
\begin{tabular}{@{}lccccc@{}}
\hline\noalign{\smallskip}
& $K$ & $I$ & time & parallel & $\hat{\mL}_K^\IW$ \\ 
\noalign{\smallskip}\hline\noalign{\smallskip}
GVA & 1 & 31 & 13.9 & - & 3138.3 (1.8) \\
CSGVA & 1 & 39 & 16.2 & - & 3139.2 (1.5) \\
\multirow{3}{*}{IW-CSGVA} & 5 & 1 & 2.5 & 6.1 & 3139.9 (0.7) \\
 & 20 & 1 & 6.9 & 8.1 & 3140.1 (0.4) \\
 & 100 & 1 & 33.5 & 16.0 & 3140.1 (0.3) \\
\noalign{\smallskip}\hline
\end{tabular}
\end{table}

The estimated marginal posterior distributions of the global parameters are shown in Figure \ref{fig_epil}. The plots show that CSGVA (red) produces improved estimates of the posterior distribution as compared to GVA (blue), especially in estimating the posterior variance of the precision parameters $\omega_2$ and $\omega_3$. The posteriors estimated using IW-CSGVA for the different values of $K$ are very close. By using just $K=5$, we are able to obtain estimates that are virtually indistinguishable from that of MCMC.
\begin{figure}[htb!]
\centering
\includegraphics[width=0.49\textwidth]{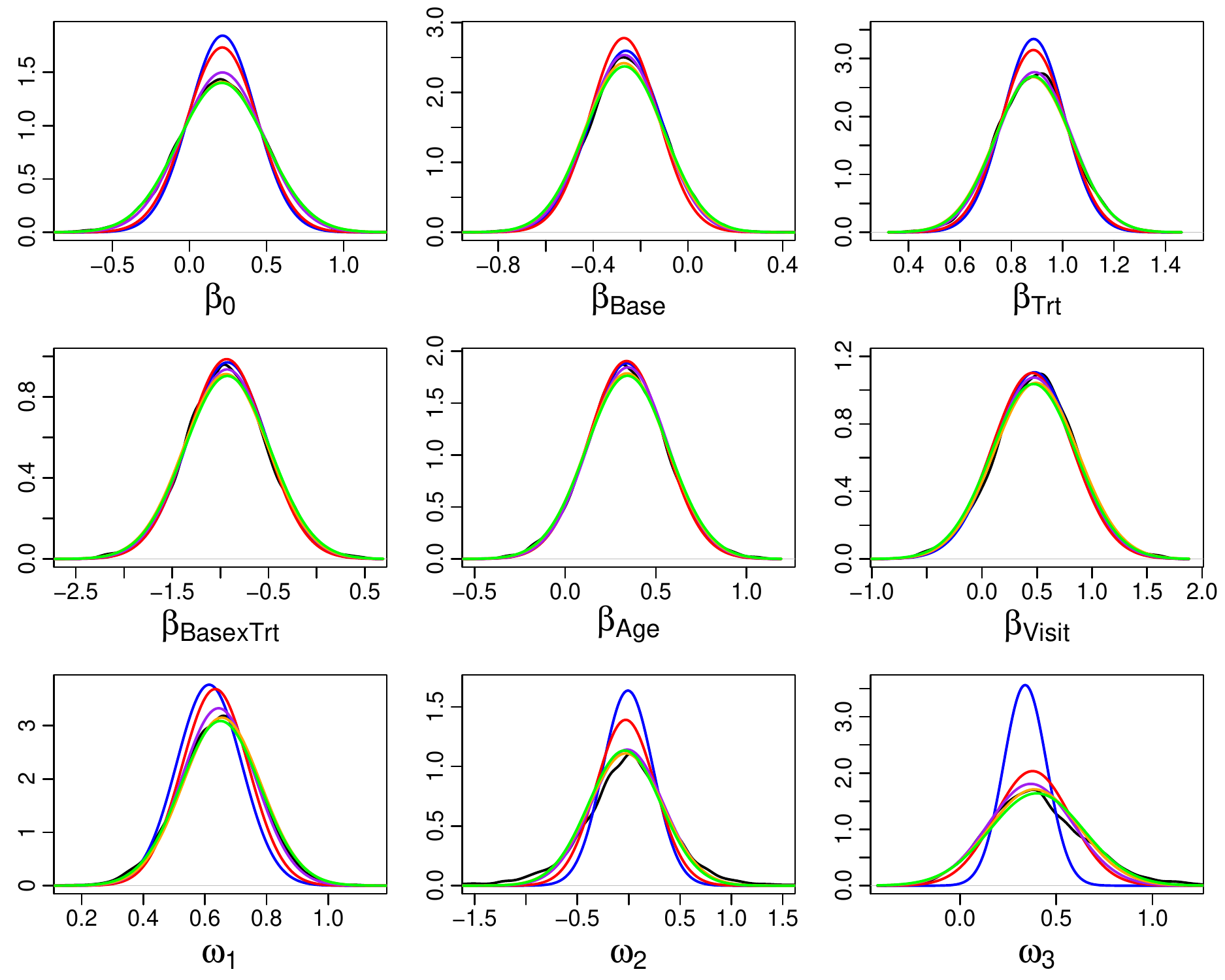}
\caption{Epilepsy data. Marginal posterior distributions of global parameters. Black (MCMC), blue (GVA), red (CSGVA), purple ($K=5$), orange ($K=20$), green ($K=100$).} \label{fig_epil}
\end{figure}

\subsection{Madras data}
These data come from the Madras longitudinal schizophrenia study \citep{Diggle2002} for detecting a psychiatric symptom called ``thought disorder". Monthly records showing whether the symptom is present in a patient are kept for $n = 86$ patients over 12 months. The response $y_{ij}$ is a binary indicator for presence of the symptom. Covariates include the time in months since initial hospitalization ($t$), gender of patient (Gender = 0 if male and 1 if female) and age of patient (Age = 0 if the patient is younger than 20 years and 1 otherwise). Consider $y_{ij} \sim \text{Bern}(\mu_{ij})$ and
\begin{multline*}
\text{logit}(\mu_{ij}) = \beta_0 +\beta_{\text{Age}}\text{ Age}_i +\beta_{\text{Gender}}\text{ Gender}_i + \beta_tt_{ij}\\
+ \beta_{\text{Age} \times t}\text{ Age}_i \times t_{ij} + \beta_{\text{Gender} \times t}\text{ Gender}_i \times t_{ij} +b_i,
\end{multline*}
for $i =1,\dots,86, 1 \leq j \leq12$.

The results in Table \ref{tab_madras} are quite similar to that of the epilepsy data. CSGVA yields an improvement in the lower bound estimate as compared to GVA and IW-CSGVA provided further improvements, which are evident starting with a $K$ as small as 5. Parallel computing halved the computation time for $K=100$ but did not yield any benefits for $K \in\{5,20\}$. From Figure \ref{fig_madras}, CSGVA and IW-CSGVA are again able to capture the posterior variance of the precision parameter $\omega_1$ better than GVA.

\begin{table}[htb!]
\caption{Madras data. Number of iterations $I$ (in thousands), runtimes (in seconds) and estimates of lower bound (standard deviation in brackets) of the variational methods.}
\label{tab_madras}
\centering 
\begin{tabular}{@{}lccccc@{}}
\hline\noalign{\smallskip}
& $K$ & $I$ & time & parallel & $\hat{\mL}_K^\IW$ \\ 
\noalign{\smallskip}\hline\noalign{\smallskip}
GVA & 1 & 25 & 13.1 & - & -383.4 (1.4) \\
CSGVA & 1 & 35 & 12.6 & - & -383.1 (1.4) \\
\multirow{3}{*}{IW-CSGVA}  & 5 & 1 & 2.4 & 7.1 & -382.5 (0.7) \\
 & 20 & 1 & 6.8 & 8.9 & -382.4 (0.4) \\
 & 100 & 1 & 33.9 & 16.8 & -382.3 (0.2) \\
\noalign{\smallskip}\hline
\end{tabular}
\end{table}

\begin{figure}[htb!]
\centering
\includegraphics[width=0.49\textwidth]{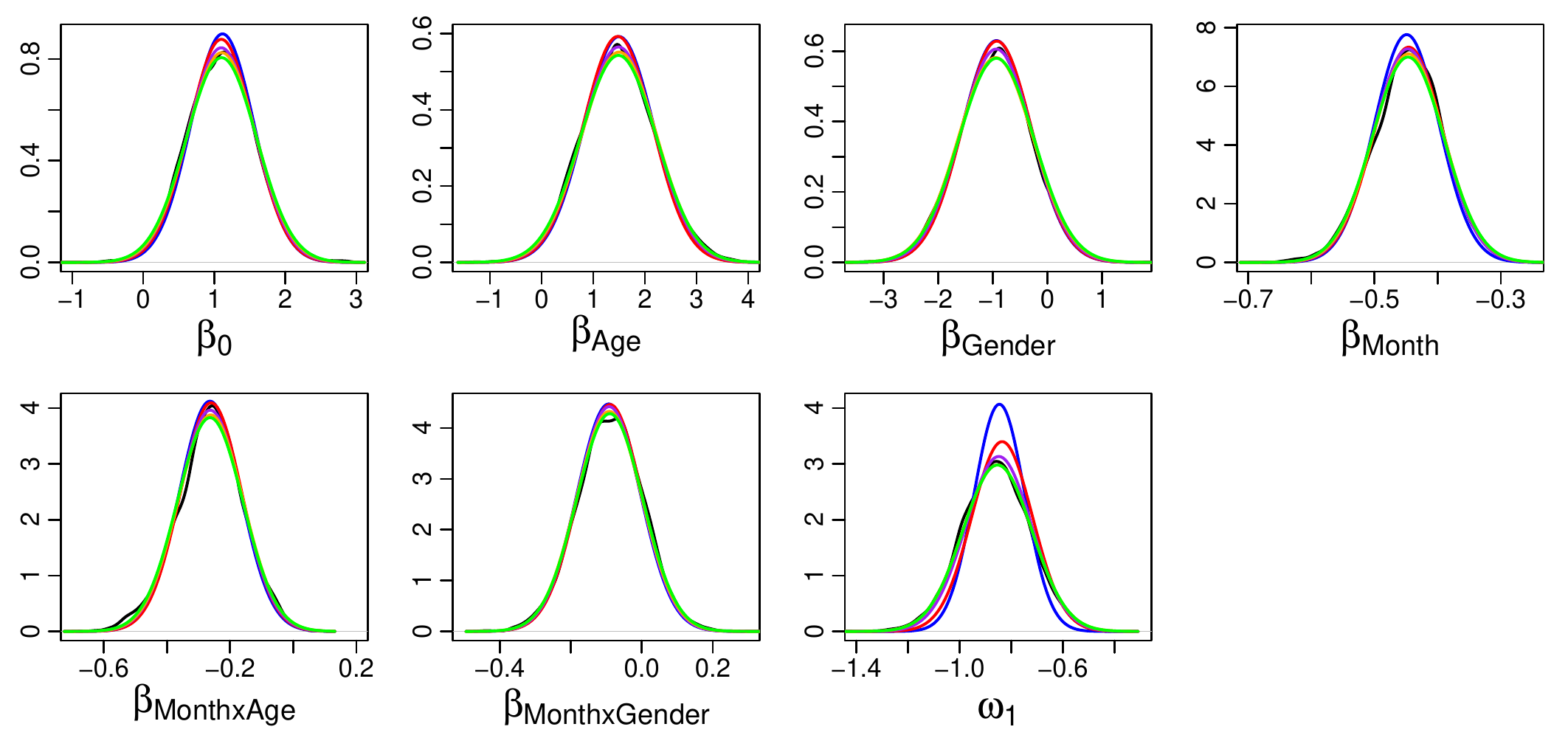}
\caption{Madras data. Marginal posterior distributions of global parameters. Black (MCMC), blue (GVA), red (CSGVA), purple ($K=5$), orange ($K=20$), green ($K=100$).} \label{fig_madras}
\end{figure}

\subsection{Six cities data}
In the six cities data \citep{Fitzmaurice1993}, $n = 537$ children from Steubenville, Ohio, are involved in a longitudinal study to investigate the health effects of air pollution. Each child is examined yearly from age 7 to 10 and the response $y_{ij}$ is a binary indicator for wheezing. There are two covariates, $\text{Smoke}_i$ (a binary indicator for smoking status of the mother of child $i$) and $\text{Age}_{ij}$ (age of child $i$ at time point $j$, centered at 9 years). Consider $y_{ij} \sim \text{Bern}(\mu_{ij})$, where
\begin{equation*}
\begin{aligned}
\text{logit}(\mu_{ij}) &= \beta_0 +\beta_{\text{Smoke}}\text{Smoke}_i +\beta_{\text{Age}} \text{ Age}_{ij} \\
&\quad + \beta_{\text{Smoke} \times \text{Age}} \text{ Smoke}_i \times \text{Age}_{ij} +b_{i},
\end{aligned}
\end{equation*}
for $i = 1,\dots, 537, j = 1,\dots, 4$.

\begin{table}[htb!]
\caption{Six cities data. Number of iterations $I$ (in thousands), runtimes (in seconds) and estimates of lower bound (standard deviation in brackets) of the variational methods.}
\label{tab_wheeze}
\centering 
\begin{tabular}{@{}lccccc@{}}
\hline\noalign{\smallskip}
& $K$ & $I$ & time & parallel & $\hat{\mL}_K$ \\ 
\noalign{\smallskip}\hline\noalign{\smallskip}
GVA & 1 & 26 & 60.3 & - & -816.4 (4.0) \\
CSGVA & 1 & 28 & 36.5 & - & -816.0 (3.9) \\
\multirow{3}{*}{IW-CSGVA}  & 5 & 1 & 6.5 & 16.3 & -812.6 (2.5) \\
& 20 & 1 & 23.1 & 24.5 & -811.0 (1.9) \\
& 100 & 1 & 115.5 & 61.4 & -809.8 (1.5) \\
\noalign{\smallskip}\hline
\end{tabular}
\end{table}

From Table \ref{tab_wheeze}, CSGVA managed to achieve a higher lower bound than GVA in about half the runtime. As $K$ increases, IW-CSGVA produced tighter lower bounds for the log marginal likelihood. As in the previous two examples, parallel computing is beneficial only when $K=100$, cutting the runtime by about half. From Figure \ref{fig_wheeze}, there is slight overestimation of the posterior means of $\beta_0$ and $\omega_1$ by all the variational methods. However, CSGVA and IW-CSGVA are able to capture the posterior variance of these parameters much better than GVA especially for $\omega_1$. 

\begin{figure}[htb!]
\centering
\includegraphics[width=0.49\textwidth]{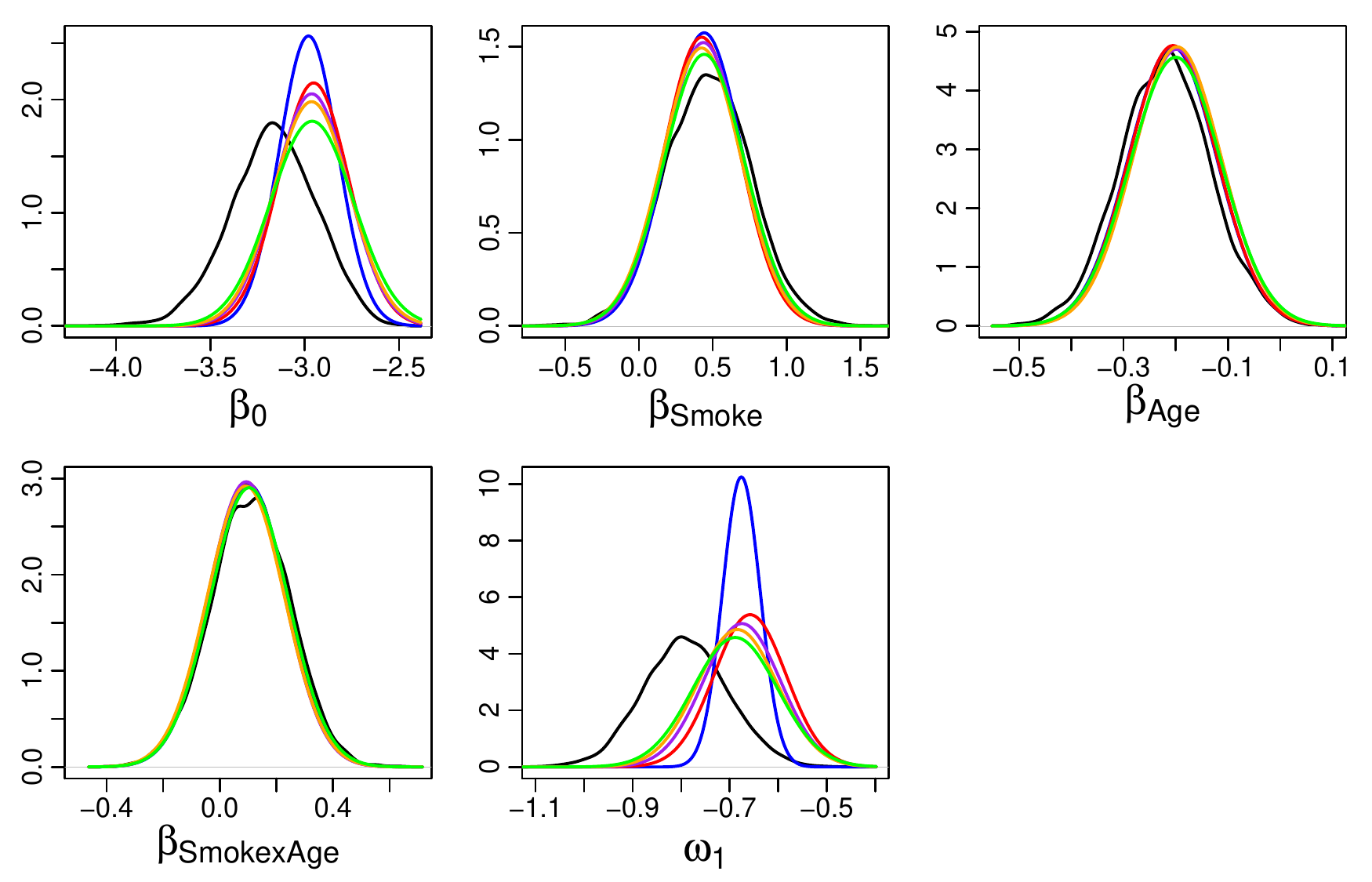}
\caption{Six cities data. Marginal posterior distributions of global parameters. Black (MCMC), blue (GVA), red (CSGVA), purple ($K=5$), orange ($K=20$), green ($K=100$).} \label{fig_wheeze}
\end{figure}

\section{Application to state space models} \label{sec_SSMs}

Here we consider the stochastic volatility model (SVM) widely used for modeling financial time series. Let each observation $y_i$ for $i=1, \dots, n$, be generated from a zero-mean Gaussian distribution where the error variance is stochastically evolving over time. The unobserved log volatility $b_i$ is modeled using an autoregressive process of order one with Gaussian disturbances: 
\begin{equation*}
\begin{gathered}
y_i|b_i, \sigma, \kappa \sim N( 0, \e^{\sigma b_i + \kappa} ), \quad (i=1,\dots,n)\\
b_i|\phi \sim N(\phi b_{i-1}, 1), \quad (i=2,\dots,n) \\
b_1|\phi \sim N(0, 1/(1-\phi^2) ),
\end{gathered}
\end{equation*}
where $\kappa \in \mathbb{R}$, $\sigma > 0$ and  $0< \phi < 1$. Here, $y_i$ represents the mean-corrected return at time $i$ and the states $\{b_i\}$ come from a stationary process with $b_1$ drawn from the stationary distribution. The parametrization of the SVM above is noncentered. The centered parametrization can be obtained by replacing $b_i$ by $(\tilde{b}_i - \kappa)/\sigma$. The performance of GVA and CSGVA are sensitive to the parametrization both in terms of rate of convergence and attained local mode. For the data sets below, the noncentered parametrization was found to have better convergence properties. The sensitivity of the stochastic volatility model to parametrization when fitted using MCMC algorithms is well known in the literature. To ``combine the best of different worlds", \cite{Kastner2014} introduce a strategy that samples parameters from the centered and noncentered parametrizations alternately. \cite{Tan2017} studies optimal partially noncentered parametrizations for Gaussian state space models fitted using EM, MCMC or  variational algorithms.

We use the following transformations to map constrained parameters to  $\mathbb{R}$:
\begin{equation*}
\alpha = \log (\exp(\sigma) - 1), \quad 
\psi = \logit (\phi).
\end{equation*}
This transformation for $\alpha$ works better than $\alpha = \log(\sigma)$, which leads to erratic fluctuations in the lower bound and convergence issues more often. The global variables are $\theta_G = [\alpha, \kappa, \psi]^T$ of dimension $G = 3$ and the local variables are $\theta_L = [b_1, \dots, b_n]^T$ of length $n$. We assume normal priors for the global parameters, where $\alpha \sim N (0, \sigma_{\alpha}^2)$, $\kappa \sim N (0, \sigma_{\kappa}^2)$ and $\psi \sim N (0, \sigma_{\psi}^2)$, where $\sigma_\alpha^2 = \sigma_\kappa^2 = \sigma_\psi^2 = 10$. The joint density can be written as
\begin{equation*}
\begin{aligned}
p(y,\theta) &= p(\alpha)p(\kappa)p(\psi)p(b_1|\psi) \bigg\{ \prod_{i=2}^{n}p(b_i|b_{i-1}, \psi) \bigg\} \\
&\times  \bigg\{ \prod^n_{i=1}p(y_i|b_i, \alpha, \kappa) \bigg\} .
\end{aligned}
\end{equation*}
The log joint density is
\begin{equation*}
\begin{aligned}
\log p(y, \theta) &= - \frac{n\kappa}{2} - \frac{\sigma}{2} \sum_{i=1}^n b_i - \frac{1}{2}\sum^n_{i=1} y_i^2 \e^{-\sigma b_i - \kappa} \\
& \quad - \frac{1}{2}\sum_{i=2}^{n} (b_i - \phi b_{i-1})^2 - \frac{1}{2} b_1^2 (1-\phi^2) \\
& \quad   + \frac{1}{2}\log(1-\phi^2) - \frac{\alpha^2}{2\sigma_{\alpha}^2} - \frac{\kappa^2}{2\sigma_{\kappa}^2} - \frac{\psi^2}{2\sigma_{\psi}^2} + c,
\end{aligned}
\end{equation*}
where $\phi = \exp(\psi)/\{1+\exp(\psi)\}$, $\sigma = \log(\exp(\alpha)+1)$ and $c$ is a constant independent of $\theta$. The gradient $\nabla_\theta \log p(y, \theta)$ is given in Appendix C. For this model, $b_i$ is conditionally independent of $b_j$ in the posterior if $ |i - j| >1$ given $\theta_G$. Thus, the sparsity structure imposed on $\Omega_2$ and $C_2$ are
\begin{equation*}
\begin{aligned}
\Omega = \begin{bmatrix}
\Omega_{2,11} & \Omega_{2,12} &0 &\dots &0	\\
 \Omega_{2,21} & \Omega_{2,22} & \Omega_{2,23} & \dots  & 0	\\
 0 & \Omega_{2,32} & \Omega_{2,33} & \dots & 0	\\
\vdots & \vdots & \vdots & \ddots & \vdots \\
0 & 0 & 0 & \dots & \Omega_{2,nn}
\end{bmatrix}, \\
C_2 = \begin{bmatrix}
C_{2,{11}} & 0 &0 &\dots  &0	\\
C_{2,21} & C_{2,{22}} & 0 & \dots  & 0	\\
 0 & C_{2,{32}} & C_{2,{33}} & \dots  & 0	\\
\vdots & \vdots & \vdots & \ddots  & \vdots \\
0 & 0 & 0 & \dots & C_{2,{nn}}
\end{bmatrix}.
\end{aligned}
\end{equation*}
The number of nonzero elements in $\vech(C_2^*)$ is $2n-1$. Setting $f_i = 0$ and $F_{ij} =0$ for all $i \in \mathcal{I}$ and all $j$, where $\mathcal{I}$ denotes the set of indices in $\vech(C_2^*)$ which are fixed as zero, the number of variational parameters to be optimized are reduced from $n(n+1)/2$ to $2n-1$ for $f$, and from $n(n+1)G/2$ to $(2n-1)G$ for $F$.

\subsection{GBP/USD exchange rate data}
This data contain 946 observations of the exchange rates of the US Dollar (USD) against the British Pound (GBP), recorded daily from 1 October 1981, to 28 June 1985. This data are available from the ``Garch" dataset in the R package {\tt Ecdat}. The mean-corrected responses $\{y_t|t=1, \dots, n\}$ are computed from the exchange rates $\{r_t|t=0, \dots, n\}$ as
\begin{equation*}
y_t = 100 \left\{\log\left(\frac{r_t}{r_{t-1}}\right) - \frac{1}{n}\sum^n_{i=1}\log\left(\frac{r_i}{r_{i-1}}\right)\right\},
\end{equation*}
where $n=945$.

For this data, CSGVA failed to achieve a higher lower bound when it was initialized using $\lambda=0$. Hence we initialize CSGVA using the fit from GVA instead, by using the association discussed in Section \ref{sec_assoc}. With this informative starting point, CSGVA converged in 16000 iterations and managed to improve upon the GVA fit, attaining a higher lower bound. IW-CSGVA led to further improvements with increasing $K$. As this dataset contains a large number of observations with correspondingly more variational parameters to be optimized, the computation is more intensive and parallel computing comes in very useful reducing the computation time by factors of 1.8, 2.9 and 4.2 for $K =5, 20, 100$ respectively. 
\begin{table}[htb!]
\caption{GBP data. Number of iterations $I$ (in thousands), runtimes (in seconds) and estimates of lower bound (standard deviation in brackets) of the variational methods.}
\label{tab_gbp}
\centering 
\begin{tabular}{@{}lccccc@{}}
\hline\noalign{\smallskip}
& $K$ & $I$ & time & parallel & $\hat{\mL}_K^\IW$ \\ 
\noalign{\smallskip}\hline\noalign{\smallskip}
GVA & 1 & 61 & 239.7 & - & -138.2 (1.3) \\
% CSGVA & 1 & 60 & 224.4 & - & -138.3 (1.5) \\
CSGVA & 1 & 16 & 58.6 & - & -137.8 (1.3) \\
\multirow{3}{*}{IW-CSGVA} & 5 & 1 & 18.3 & 10.2  & -137.4 (1.0) \\
& 20 & 1 & 71.2 & 24.4 & -137.0 (0.5) \\
& 100 & 1 & 355.3 & 84.3 & -136.8 (0.4) \\
\noalign{\smallskip}\hline
\end{tabular}
\end{table}

Figure \ref{fig1_gbp} shows the estimated marginal posteriors of the global parameters. CSGVA provides significant improvements in estimating the posterior variance of $\alpha$ and $\psi$ as compared to GVA. With the application of IW-CSGVA, we are able to obtain posterior estimates that are quite close to that of MCMC even with a small $K$.
\begin{figure}[htb!]
\centering
\includegraphics[width=0.49\textwidth]{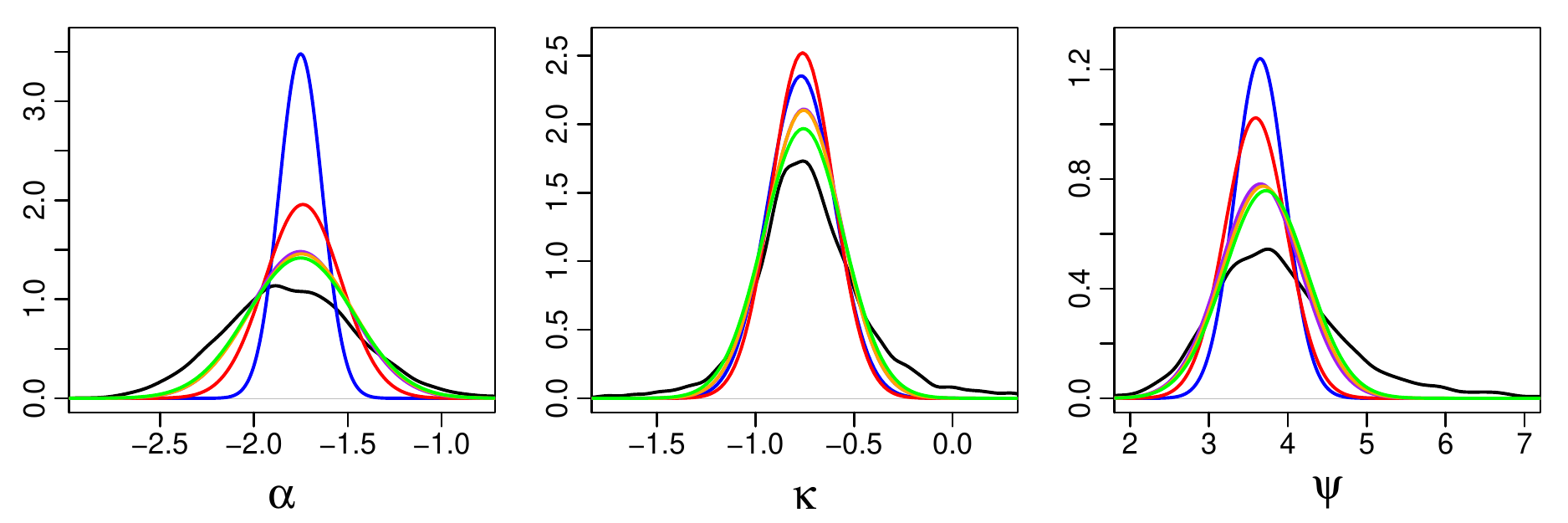}
\caption{GBP data. Marginal posterior distributions of global parameters. Black (MCMC), blue (GVA), red (CSGVA), purple ($K=5$), orange ($K=20$), green ($K=100$).} \label{fig1_gbp}
\end{figure}

Figure \ref{fig2_gbp} shows the estimated marginal posteriors of the latent states $\{b_i\}$ summarized using the mean (solid line) and one standard deviation from the mean (dotted line) estimated by MCMC (black) and IW-CSGVA ($K=5$, purple). 
\begin{figure}[htb!]
\centering
\includegraphics[width=0.49\textwidth]{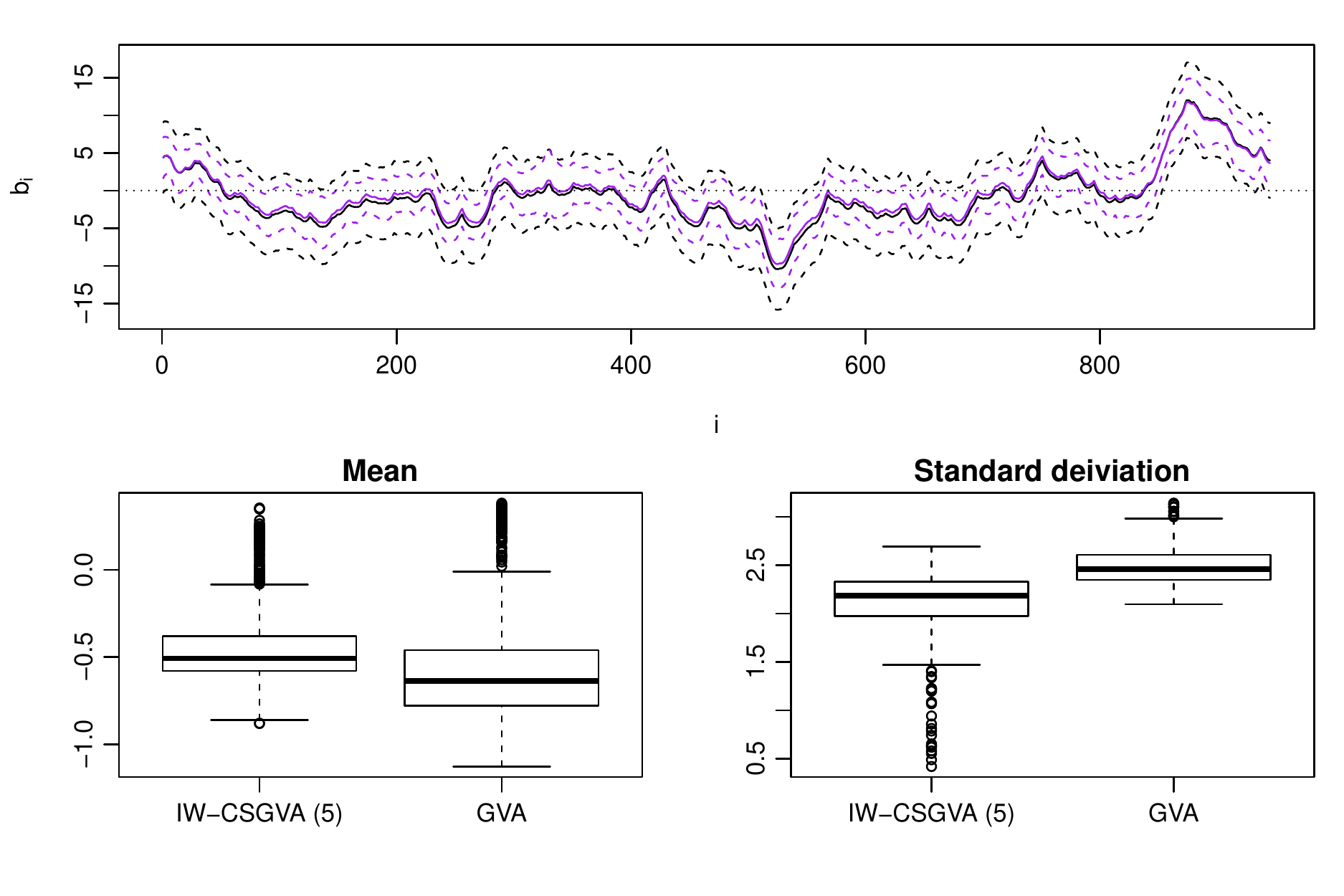}
\caption{GBP data. Top: Posterior means (solid line) of the latent states and one standard deviation from the mean (dotted line) estimated using MCMC (black) and IW-CSGVA ($K=5$, purple). Bottom: Boxplots of $\text{mean}_{\text{MCMC}} - \text{mean}_{\text{VA}}$ and $\text{sd}_{\text{MCMC}} - \text{sd}_{\text{VA}}$.} \label{fig2_gbp}
\end{figure}
For IW-CSGVA, the means and standard deviations are estimated based on 2000 samples, by generating $\theta_G$ from $q(\theta_G)$ followed by $\theta_L$ from $q(\theta_L|\theta_G)$. For MCMC, estimation was based on 5000 samples. IW-CSGVA estimated the means quite accurately (with a little overestimation) but the standard deviations are underestimated when compared to MCMC. The boxplots shows the difference between the means and standard deviations estimated by IW-CSGVA $(K=5)$ and GVA with MCMC. We can see that IW-CSGVA estimated both the means and standard deviations more accurately as compared to GVA. 

\subsection{New York stock exchange data}
Next we consider 2000 observations of the returns of the New York Stock Exchange (NYSE) from February 2, 1984 to December 31, 1991. This data is available as the dataset ``nyse" from the {\tt R} package {\tt astsa}. We consider 100 times the mean-corrected returns as responses $\{y_i\}$.

From Table \ref{tab_nyse}, CSGVA was able to attain a higher lower bound than GVA when initialized in the standard manner using $\lambda=0$. Applying IW-CSGVA led to further improvements as $K$ increases. For this massive data set, parallel computing yields significant reductions in computation time, by factors of 2.9, 4.5 and 5.5 corresponding to $K=5$, 20, 100 respectively.
\begin{table}[htb!]
\caption{NYSE data. Number of iterations $I$ (in thousands), runtimes (in seconds) and estimates of lower bound (standard deviation in brackets) of the variational methods.}
\label{tab_nyse}
\centering
\begin{tabular}{@{}lccccc@{}}
\hline\noalign{\smallskip}
& $K$ & $I$ & time & parallel & $\hat{\mL}_K^\IW$ \\  
\noalign{\smallskip}\hline\noalign{\smallskip}
GVA & 1 & 43 & 679.0 & - & -570.8 (1.8) \\
CSGVA & 1 & 49 & 749.2 & - & -570.7 (2.0) \\
% CSGVA & 1 & 22 & 334.8 & - & -570.7 (2.1) \\
\multirow{3}{*}{IW-CSGVA} & 5 & 1 & 76.0 & 26.1  & -569.4 (1.1) \\
& 20 & 1 & 305.0 & 67.9 & -569.0 (0.7) \\
& 100 & 1 & 1503.0 & 274.0 & -568.7 (0.4) \\
\noalign{\smallskip}\hline
\end{tabular}
\end{table}

Figure \ref{fig1_nyse} shows that the marginal posteriors estimated using CSGVA are quite close to that of MCMC while GVA underestimated the posterior variance of $\alpha$ and $\psi$ quite severely. Posteriors estimated by IW-CSGVA are virtually indistinguishable from MCMC.
\begin{figure}[htb!]
\centering
\includegraphics[width=0.49\textwidth]{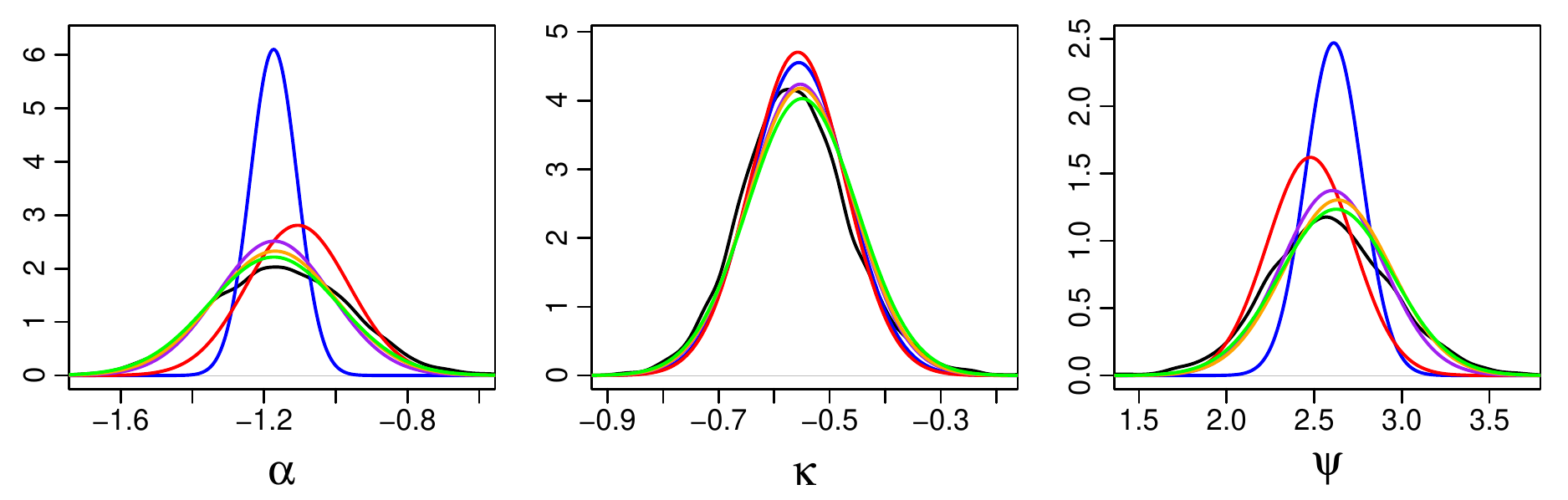}
\caption{NYSE data. Marginal posterior distributions of global parameters. Black (MCMC), blue (GVA), red (CSGVA), purple ($K=5$), orange ($K=20$), green ($K=100$).} \label{fig1_nyse}
\end{figure}

From Figure \ref{fig2_nyse}, we can see that the marginal posteriors of the latent states are also estimated very well using IW-CSGVA $(K=5)$ and there is no systematic underestimation of the posterior variance unlike the previous example. GVA captures the posterior means very well but did not perform as well as IW-CSGVA in estimating the posterior variance.
\begin{figure}[htb!]
\centering
\includegraphics[width=0.49\textwidth]{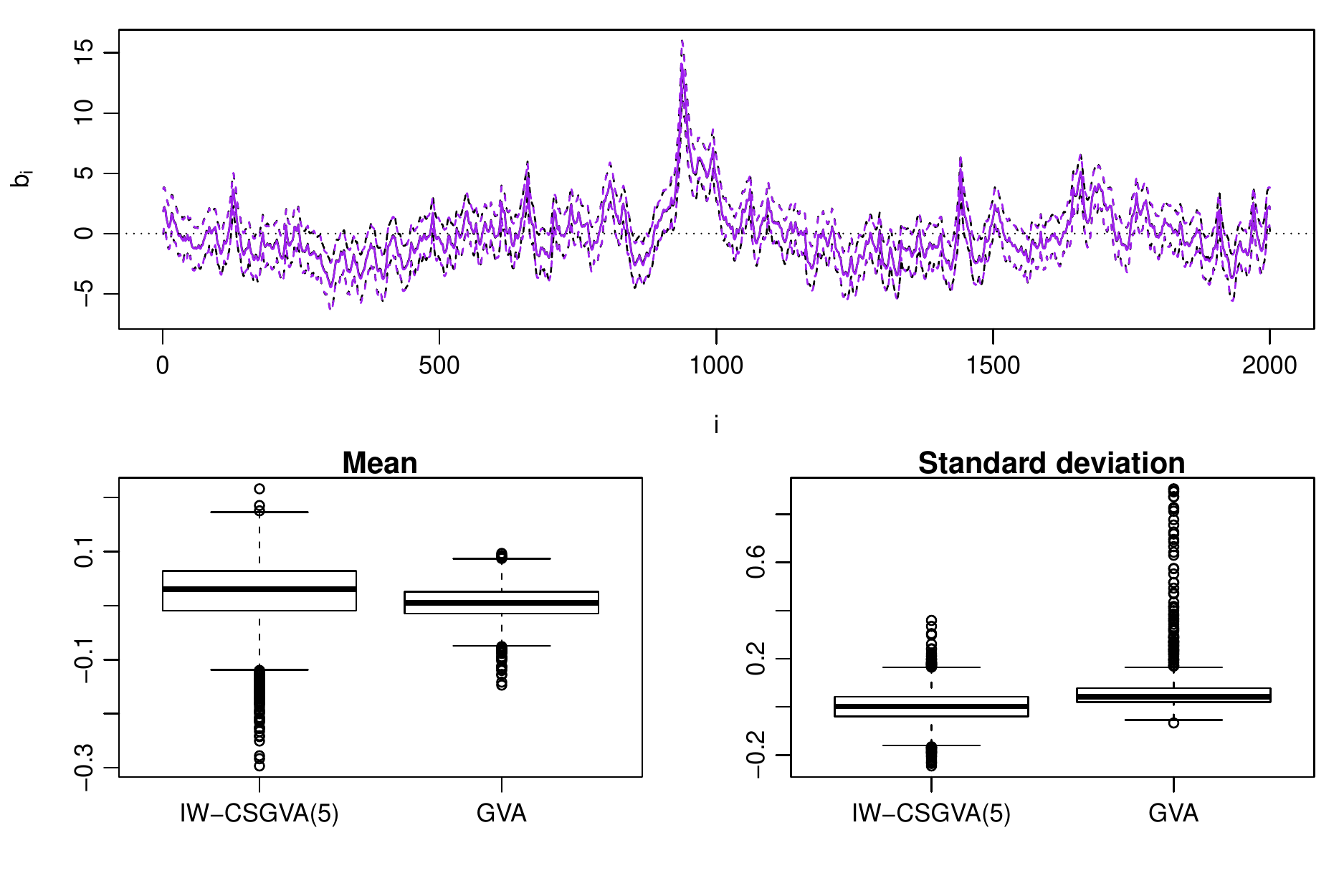}
\caption{NYSE data. Top: Posterior means (solid line) of the latent states and one standard deviation from the mean (dotted line) estimated using MCMC (black) and IW-CSGVA ($K=5$, purple). Bottom: Boxplots of $\text{mean}_{\text{MCMC}} - \text{mean}_{\text{VA}}$ and $\text{sd}_{\text{MCMC}} - \text{sd}_{\text{VA}}$.} \label{fig2_nyse}
\end{figure}

\section{Conclusion} \label{sec_conclu}
In this article, we have proposed a Gaussian variational approximation for hierarchical models that adopts a conditional structure $q(\theta) = q(\theta_G) q(\theta_L|\theta_G)$. The dependence of the local variables $\theta_L$ on global variables $\theta_G$ are then captured using a linear approximation. This structure is very useful when there are global scale parameters in $\theta_G$ which help to determine the scale of local variables in the conditional posterior of $\theta_L|\theta_G$. We further demonstrate how CSGVA can be improved by maximizing the importance weighted lower bound. From our experiments, using a $K$ as small as 5 can lead to significant improvements in the variational approximation, with just a short run. Moreover, for massive datasets, computation time can be further reduced through parallel computing. Our experiments indicate that CSGVA coupled with importance weighting is particularly useful in improving the estimation of the posterior variance of precision parameters $\omega$ in GLMMs, and the persistence $\phi$ and volatility $\sigma$ of the log-variance in SVMs.

%\begin{acknowledgements}
%\end{acknowledgements}

\bibliographystyle{spbasic}      % basic style, author-year citations
\bibliography{ref}   % name your BibTeX data base

\appendix

\section{Derivation of stochastic gradient}
We have 
\begin{equation*}
r_\lambda(s) = \begin{bmatrix} \theta_G \\ \theta_L \end{bmatrix} 
= \begin{bmatrix}
\mu_1 + C_1^{-T} s_1 \\ 
d+ C_2^{-T} (s_2 - DC_1^{-T} s_1)
\end{bmatrix},
\end{equation*}
where $\vech(C_2^*)  = f + F(\mu_1 + C_1^{-T} s_1)$. Differentiating $r_\lambda(s)$ with respect to $\lambda$, $\nabla_\lambda r_\lambda (s) $ is given by

\begin{equation*}
\begin{aligned}
&\begin{bmatrix}
\nabla_{\mu_1} \theta_G & \nabla_{\mu_1} \theta_L \\ 
\nabla_{\vech(C_1^*)} \theta_G & \nabla_{\vech(C_1^*)} \theta_L \\ 
\nabla_{d} \theta_G & \nabla_{d} \theta_L \\ 
\nabla_{\VEC(D)} \theta_G & \nabla_{\VEC(D)} \theta_L \\ 
\nabla_{f} \theta_G & \nabla_{f} \theta_L \\ 
\nabla_{\VEC(F)} \theta_G & \nabla_{\VEC(F)} \theta_L \\ 
\end{bmatrix}.
\end{aligned}
\end{equation*}
Since $\theta_G$ does not depend on $d$, $D$, $f$ and $F$, we have
\begin{equation*}
\begin{gathered}
\nabla_{d} \theta_G = 0_{nL \times G}, \quad \nabla_{\VEC(D)} \theta_G = 0_{nLG \times G} \\
\nabla_{f} \theta_G = 0_{nL(nL+1)/2 \times G}, \;\; \nabla_{\VEC(F)} \theta_G = 0_{nLG(nL+1)/2 \times G}.
\end{gathered}
\end{equation*}
It is easy to see that $\nabla_{\mu_1} \theta_G = I_G$ and $\nabla_d \theta_L = I_{nL}$. The rest of the terms are derived as follows. 

Differentiating $\theta_G$ with respect to $\vech(C_1^*)$,
\begin{equation*}
\begin{aligned}
\diff \theta_G &= -C_1^{-T} \diff(C_1^T) C_1^{-T} s_1\\
&= - (s_1^T C_1^{-1} \otimes C_1^{-T}) K_G E_G^T D_1^* \diff \vech(C_1^*) \\
&= - (C_1^{-T} \otimes s_1^T C_1^{-1}) E_G^T D_1^* \diff \vech(C_1^*).
\end{aligned}
\end{equation*}
\begin{equation*}
\begin{aligned}
\therefore \; \nabla_{\vech(C_1^*)} \theta_G & = - D_1^* E_G (C_1^{-1} \otimes C_1^{-T} s_1 ).
\end{aligned}
\end{equation*}

Differentiating $\theta_L$ with respect to $f$,
\begin{equation*}
\begin{aligned}
\diff \theta_L &= -C_2^{-T} \diff (C_2^T) C_2^{-T} (s_2 - DC_1^{-T} s_1) \\
&= - \{ (s_2 - DC_1^{-T} s_1)^T C_2^{-1} \otimes C_2^{-T} \} \\
& \quad \times K_{nL} E_{nL}^T D_2^* \diff f 
\end{aligned}
\end{equation*}
\begin{equation*}
\therefore \; \nabla_{f} \theta_L = - D_2^* E_{nL} \{ C_2^{-1} \otimes C_2^{-T} (s_2 - DC_1^{-T} s_1) \}.
\end{equation*}

Differentiating $\theta_L$ with respect to $F$, 
\begin{equation*}
\begin{aligned}
\diff \theta_L &= (\nabla_{f} \theta_L)^T \diff F \theta_G \\
&= \{ \theta_G^T \otimes (\nabla_{f} \theta_L)^T \} \diff \VEC(F). 
\end{aligned}
\end{equation*}
\begin{equation*}
\therefore \; \nabla_{\VEC(F)} \theta_L = \theta_G \otimes \nabla_{f} \theta_L.
\end{equation*}

Differentiating $\theta_L$ with respect to $D$, 
\begin{equation*}
\begin{aligned}
\diff \theta_L &= -C_2^{-T} \diff D C_1^{-T} s_1 \\
&= - (s_1^T C_1^{-1} \otimes C_2^{-T}) \diff \VEC(D).
\end{aligned}
\end{equation*}
\begin{equation*}
\therefore \; \nabla_{\VEC(D)} \theta_L =  - (C_1^{-T} s_1 \otimes C_2^{-1}).
\end{equation*}

Differentiating $\theta_L$ with respect to $\mu_1$, 
\begin{equation*}
\begin{aligned}
\diff \theta_L &= (\nabla_{f} \theta_L)^T F \diff \mu_1 \\ 
\end{aligned}
\end{equation*}
\begin{equation*}
\therefore \;  \nabla_{\mu_1} \theta_L = F^T (\nabla_{f} \theta_L).
\end{equation*}

Differentiating $\theta_L$ with respect to $\vech(C_1)$, 
\begin{equation*}
\begin{aligned}
\diff \theta_L &= -C_2^{-T}\diff (C_2^T) C_2^{-T}(s_2 - DC_1^{-T} s_1) \\
& \quad - C_2^{-T} D \diff (C_1^{-T}) s_1 \\
& = (\nabla_{f} \theta_L)^T F\diff (C_1^{-T}) s_1 - C_2^{-T} D \diff (C_1^{-T}) s_1 \\
& = \{  (\nabla_{f} \theta_L)^T F -  C_2^{-T} D\} (\nabla_{\vech(C_1^*)} \theta_G )^T \diff \vech(C_1^*) \\
\end{aligned}
\end{equation*}
\begin{equation*}
\begin{aligned}
\therefore \; \nabla_{\vech(C_1^*)} \theta_L &= \nabla_{\vech(C_1^*)}\theta_G  \{  F^T \nabla_{f} \theta_L -  D^T C_2^{-1} \}  \\
&= \nabla_{\vech(C_1^*)}\theta_G  \{  \nabla_{\mu_1} \theta_L -  D^T C_2^{-1} \}.
\end{aligned}
\end{equation*}

Since $s_1 = C_1^T(\theta_G - \mu_1)$ and $s_2 = C_2^T (\theta_L - \mu_2)$, we have
\begin{equation*}
\begin{aligned}
&\log q_\lambda (\theta) = \log q(\theta_G) + \log q(\theta_L|\theta_G) \\
&= -\frac{G}{2} \log(2\pi) + \log |C_1| - \frac{1}{2}(\theta_G - \mu_1)^T C_1 C_1^T (\theta_G - \mu_1) \\
& \quad -\frac{nL}{2} \log(2\pi) + \log |C_2| - \frac{1}{2}(\theta_L - \mu_2)^T C_2 C_2^T (\theta_L - \mu_2) \\
& = - \frac{nL+G}{2} \log (2\pi) + \log |C_1C_2| - \frac{1}{2} s^T s.
\end{aligned}
\end{equation*}

As $\mu_2 = d + C_2^{-T} D(\mu_1 - \theta_G)$ and $\vech(C_2^*) = f + F \theta_G$, differentiating $\log q_\lambda (\theta)$ with respect to $\theta_G$,
\begin{equation*}
\begin{aligned}
&\diff \log q_\lambda (\theta) \\
&= - (\theta_G - \mu_1)^T C_1 C_1^T \diff \theta_G 
-  (\theta_L - \mu_2)^T C_2 C_2^T(-\diff \mu_2)\\
& \quad - (\theta_L - \mu_2)^T \diff C_2 s_2 + \tr(C_2^{-1} \diff C_2) \\
&= -s_1^T C_1^T \diff \theta_G + s_2^T C_2^T\{- C_2^{-T} D \diff \theta_G + \diff (C_2^{-T}) D(\mu_1 - \theta_G)\} \\
& \quad - \VEC(C_2^{-T} s_2 s_2^T)^T \diff \VEC(C_2) + \VEC(C_2^{-T})^T d\VEC(C_2) \\
&= \VEC(C_2^{-T} - \{C_2^{-T} s_2 + (\mu_2 -d)\}s_2^T)^T d\VEC(C_2) \\
& \quad -s_1^T C_1^T \diff \theta_G - s_2^T D \diff \theta_G  \\
&= \VEC(C_2^{-T} - (\theta_L -d) s_2^T)^T E_{nL}^T D_2^* F \diff \theta_G  \\
& \quad -s_1^T C_1^T \diff \theta_G - s_2^T D \diff \theta_G.
\end{aligned}
\end{equation*}
Therefore
\begin{equation*}
\begin{aligned}
\nabla_{\theta_G} \log q_\lambda (\theta) 
&=F^T D_2^* \vech(C_2^{-T} - (\theta_L -d) s_2^T) \\
& \quad - C_1 s_1 - D^T s_2.
\end{aligned}
\end{equation*}
Note that $D_2^* \vech(C_2^{-T}) = \vech(I_{nL})$ as $C_2^{-T}$ is upper triangular and $\vech(C_2^{-T})$ only retains the diagonal elements of $C_2^{-T}$.

Differentiating $\log q_\lambda (\theta)$ with respect to $\theta_L$,
\begin{equation*}
\begin{aligned}
\diff \log q_\lambda (\theta) & = - (\theta_L - \mu_2)^T C_2 C_2^T \diff \theta_L \\
& = - s_2^T C_2^T \diff \theta_L.
\end{aligned}
\end{equation*}
\begin{equation*}
\therefore \; \nabla_{\theta_L} \log q_\lambda (\theta) = - C_2 s_2.
\end{equation*}

\section{Gradients for generalized linear mixed models}
Since $\theta = [\beta^T, \omega^T, \tilde{b}_1^T, \dots, \tilde{b}_n^T]^T$, we require 
\begin{multline*}
\nabla_\theta \log p(y, \theta) = [\nabla_\beta \log p(y, \theta), \nabla_\omega \log p(y, \theta), \\
\nabla_{\tilde{b}_1} \log p(y, \theta), \dots, \nabla_{\tilde{b}_n} \log p(y, \theta)]^T.
\end{multline*}
For the centered parametrization, the components in $\nabla_\theta \log p(y, \theta)$ are given below. Note that $\beta = [\beta_{RG_1}^T, \beta_{G_2}^T]^T$.

\begin{equation*}
\nabla_{\beta_{G_2}} \log p(y, \theta) = \sum_{i=1}^n  {X_i^{G_2}}^T \{ y_i - h'(\eta_i) \} - \beta_{G_2}/\sigma_\beta^2.
\end{equation*}

\begin{equation*}
\nabla_{\beta_{RG_1}} \log p(y, \theta) = \sum_{i=1}^n C_i^T W W^T (\tilde{b}_i - C_i \beta_{RG_1}) - \beta_{RG_1}/\sigma_\beta^2.
\end{equation*}

Differentiating $\log p(y, \theta) $ with respect to $\omega$,
\begin{equation*}
\begin{aligned}
\diff \log p(y, \theta) & = -\sum_{i=1}^n (\tilde{b}_i - C_i \beta_{RG_1})^T \diff W W^T (\tilde{b}_i - C_i \beta_{RG_1}) \\
& \quad + n \tr(W^{-1} \diff W) - \omega^T \diff \omega /\sigma_\omega^2 \\
& = \VEC\bigg\{  - \sum_{i=1}^n (\tilde{b}_i - C_i \beta_{RG_1}) (\tilde{b}_i - C_i \beta_{RG_1})^T W  \\
& \quad + nW^{-T}  \bigg\}^T  E_L^T D_L^* \diff \omega - \omega^T \diff \omega /\sigma_\omega^2,
\end{aligned}
\end{equation*}
where $\diff \vech(W) = D^*_L \diff \omega$ and $D^*_L = \diag\{ \vech(\dg(W) + \bone_L\bone_L^T - I_L) \}$. Hence
\begin{equation*}
\begin{aligned}
\nabla_\omega \log p(y, \theta) 
&=- D^*_L \sum_{i=1}^n  \vech\{ (\tilde{b}_i - C_i \beta_{RG_1}) (\tilde{b}_i - C_i \beta_{RG_1}) ^TW \}   \\
& \quad + n \vech(I_L) - \omega/\sigma_\omega^2.
\end{aligned}
\end{equation*}
Note that $D_L^* \vech(W^{-T}) = \vech(I_L)$ because $W^{-T}$ is upper triangular and $\vech(W^{-T})$ only retains the diagonal elements.

\begin{equation*}
\nabla_{\tilde{b}_i} \log p(y, \theta) =  Z_i^T \{ y_i - h'(\eta_i)\} - W W^T (\tilde{b}_i - C_i \beta_{RG_1}).
\end{equation*}

\section{Gradients for state space models}
Since $\theta = [\alpha, \kappa, \psi, b_1^T, \dots, b_n^T]^T$, we require 
\begin{multline*}
\nabla_\theta \log p(y, \theta) = [\nabla_\alpha \log p(y, \theta), \nabla_\kappa \log p(y, \theta), \\
\nabla_\psi \log p(y, \theta), \nabla_{b_1} \log p(y, \theta), \dots, \nabla_{b_n} \log p(y, \theta)]^T.
\end{multline*}
The components in $\nabla_\theta \log p(y, \theta)$ are given below.

\begin{equation*}
\nabla_\alpha \log p(y, \theta) = \frac{1}{2}\sum^n_{i=1} (b_i y_i^2 \e^{-\sigma b_i - \kappa} - b_i)(1-\e^{-\sigma})  - \frac{\alpha}{\sigma_{\alpha}^2}.
\end{equation*}

\begin{equation*}
\nabla_\kappa \log p(y, \theta) = \frac{1}{2} \bigg(\sum^n_{i=1} y_i^2 \e^{-\sigma b_i - \kappa} - n \bigg) - \kappa/\sigma_{\kappa}^2.
\end{equation*}

\begin{equation*}
\begin{aligned}
\nabla_\psi \log p(y, \theta) &= \bigg\{ \sum_{i=2}^{n} (b_i - \phi b_{i-1})b_{i-1}  + b_1^2 \phi  - \frac{\phi}{1-\phi^2} \bigg\} \\
& \quad \times \phi(1-\phi)  - \psi/\sigma_{\psi}^2.
\end{aligned}
\end{equation*}

\begin{equation*}
\nabla_{b_1} \log p(y, \theta) = \frac{\sigma}{2} (y_1^2 \e^{-\sigma b_1 - \kappa} - 1) + \phi (b_2 - \phi b_1)- b_1 (1-\phi)^2.
\end{equation*}

For $2 \leq i \leq n-1$,
\begin{equation*}
\begin{aligned}
\nabla_{b_i} \log p(y, \theta) &= \frac{\sigma}{2} (y_i^2 \e^{-\sigma b_i - \kappa} - 1) +\phi  (b_{i+1} - \phi b_i)\\
& \quad  -  (b_i - \phi b_{i-1}).
\end{aligned}
\end{equation*}

\begin{equation*}
\begin{aligned}
\nabla_{b_n} \log p(y, \theta) = \frac{\sigma}{2} (y_n^2 \e^{-\sigma b_n - \kappa} - 1) -  (b_n - \phi b_{n-1}).
\end{aligned}
\end{equation*}

\end{document}